\documentclass[a4paper,onecolumn,accepted=2020-08-02]{quantumarticle}
\newif\ifcccsub
\cccsubtrue

\pdfoutput=1

\usepackage[numbers,sort&compress]{natbib}
\usepackage[margin = 1 in]{geometry}
\usepackage{graphicx,float}
\usepackage{amstext,amsmath,amssymb,amsthm} 

\usepackage{xcolor}
\usepackage{hyperref}
\usepackage{authblk}
\usepackage{subcaption}
\usepackage{algorithm}
\usepackage[noend]{algpseudocode}

\usepackage{mdframed} %
\mdfdefinestyle{figstyle}{ %
  linecolor=black!7, %
  backgroundcolor=black!7, %
  innertopmargin=10pt, %
  innerleftmargin=25pt, %
  innerrightmargin=25pt, %
  innerbottommargin=10pt %
}

\theoremstyle{plain}
\newtheorem{theorem}{Theorem}[section]

\newtheorem{lemma}[theorem]{Lemma}

\newtheorem{proposition}[theorem]{Proposition}
\newtheorem{corollary}[theorem]{Corollary}
\newtheorem{definition}[theorem]{Definition}

\newtheorem{asn}{Assumption}
\newtheorem{remark}{Remark}

\newcommand*{\bra}[1]{\langle#1|}
\newcommand*{\ket}[1]{|#1\rangle}
\DeclareMathOperator{\Tr}{Tr}

\newcommand{\generator}{G}
\newcommand{\reduction}{R}
\newcommand{\cnot}{C}
\newcommand{\qx}{\hat{Q}}

\newcommand{\lang}[1]{\textsc{#1}}
\newcommand{\class}[1]{\textsf{#1}}
\newcommand{\np}{\class{NP}}
\newcommand{\bool}{\left\{0,1\right\}}
\newcommand{\red}[1]{\leq_{#1}}
\newcommand{\mypar}[1]{\vspace{.5ex}\noindent\textbf{#1.}\vspace{.5ex}}

\newcommand{\Gen}{\textsf{Gen}}
\newcommand{\owp}{\textsf{OWP}}
\newcommand{\iowp}{\textsf{Inv-OWP}}

\DeclareMathOperator{\tr}{Tr} %
\DeclareMathOperator{\poly}{poly} %
\DeclareMathOperator{\negl}{negl} %
\newcommand{\secpar}{n} %

\newcommand{\kera}[1]{\ensuremath{\ket{#1}\bra{#1}}}

\newcommand{\abs}[1]{\ensuremath{\left\lvert #1 \right\rvert}} %
\newcommand{\hilbert}{\ensuremath{\mathcal{H}}}
\newcommand{\projr}{\ensuremath{\Pi_R}}

\newcommand{\succp}{\ensuremath{p_{\text{succ}}}}
\newcommand{\accs}{\ensuremath{S_{\text{acc}}}}
\newcommand{\malp}{\ensuremath{\tilde U}}
\newcommand{\veps}{\ensuremath{\varepsilon}}

\title{On Basing One-way Permutations on NP-hard Problems
  under Quantum Reductions} %
\author[1]{Nai-Hui Chia}
\author[2]{Sean Hallgren}\thanks{Partially supported by National Science
    Foundation awards CNS-1617802 and CCF-1618287, and by the National Security
    Agency (NSA) under Army Research Office (ARO) contract number
    W911NF-12-1-0541.}
\author[3]{Fang Song}
\affil[1]{Department of Computer Science, University of Texas at Austin, Austin, TX, 78712, USA.}

\affil[2]{Department of Computer Science and
    Engineering, The Pennsylvania State University,
University Park, PA, 16802, USA.}

\affil[3]{Department of Computer Science, Portland State University, Portland, OR 97201, USA.}
\date{}

\begin{document}
\maketitle
\begin{abstract}
A fundamental pursuit in complexity theory concerns reducing worst-case problems to average-case problems.  There exist complexity classes such as PSPACE that admit worst-case to average-case reductions.  However, for many other classes such as NP, the evidence so far is typically negative, in the sense that the existence of such reductions would cause collapses of the polynomial hierarchy(PH). Basing cryptographic primitives, e.g., the average-case hardness of inverting one-way permutations, on NP-completeness is a particularly intriguing instance. As there is evidence showing that classical reductions from NP-hard problems to breaking these primitives result in PH collapses, it seems unlikely to base cryptographic primitives on NP-hard problems. Nevertheless, these results do not rule out the possibilities of the existence of quantum reductions. In this work, we initiate a study of the quantum analogues of these questions.  Aside from formalizing basic notions of quantum reductions and demonstrating powers of quantum reductions by examples of separations, our main result shows that if NP-complete problems reduce to inverting one-way permutations using certain types of quantum reductions, then  \class{coNP} $\subseteq$ \class{QIP}($2$).
\end{abstract}

\section{Introduction}

A fundamental question in complexity theory is whether or not
worst-case problems have reductions to average-case problems.  A special case,
random-self-reductions, concerns reducing worst-case problems to
average-case problems of the same type.  Random-self-reductions exist
for complete sets of some classes such as \class{PSPACE}, \class{EXP}
and \class{\#P}.  Such reductions are not known to exist for
\class{NP}-complete problems.  Fortnow and Feigenbaum~\cite{FF93}
showed that sets which are complete for any level of the polynomial
hierarchy are not (non-adaptively) random-self-reducible unless the
polynomial hierarchy collapses, giving negative evidence for this
possibility.

More broadly, one can ask when the worst-case instances of one problem
can be reduced to random instances of a different problem, including
basing cryptographic primitives on \class{NP}-complete problems.
Namely, can we reduce \class{NP}-complete problems to breaking the
security of cryptosystems? It has been shown that approximating
several worst-case lattice problems reduce to average-case lattice
problems~\cite{Ajtai96,Mic04,Mic07}. These results ignited the booming of
lattice cryptography, and breaking these cryptosystems is as hard as
solving some worst-case lattice
problems~\cite{Regev04_2,Regev05,Ajtai97}. However, the worst-case lattice problems involved are in
\class{NP}$\cap$ \class{coNP}, which are believed not
\class{NP}-complete.

The pursuit of basing cryptographic primitives on NP-hardness has
largely ended up negative.  For instance, if one can reduce
\class{NP}-complete problems to inverting one-way
permutations~\cite{Brassard79}, size-verifiable one-way
functions~\cite{Akavia2006,Bogdanov2015}, single-server single-round
private information retrieval~\cite{LV15}, or some weak fully
homomorphic encryption scheme~\cite{BL13}, then the polynomial
hierarchy collapses. In this paper we bring this question in the
quantum computing paradigm: can cryptographic primitives be based on
\class{NP}-complete or \class{QMA}-complete problems if we allow
\emph{quantum} reductions? Namely, can we establish the security of
cryptographic primitives by showing a \emph{quantum} algorithm for an
\class{NP}-hard (or \class{QMA}-hard) problem, whenever there is an
attacker breaking the cryptographic primitive?

Is this hopeful at all? First note that previous negative results
in~\cite{BT06,Akavia2006,Bogdanov2015,FF93} would fail when quantum
reductions are allowed. Specifically, the hypothetical worst-case to
average-case reduction is used to construct a multi-round interactive
protocol (\class{IP(k)} with $k\geq 3$) for the complement of an \class{NP}-complete
language. Then we are able to claim containment of \class{coNP} in \class{AM} (and hence and
collapse of the polynomial hierarchy), because of a nice result in
complexity theory that \class{IP(k)} $=$ \class{AM}~\cite{GS86}. A direct translation of the classical argument in a quantum reduction\footnote{In some cases, such translation does not exist due to some properties of quantum information}, nonetheless, would result
in a \emph{quantum interactive proof} protocol with multiple
rounds. We would conclude that \class{coNP} $\subseteq$ \class{QIP}(k). But this is
trivially true, since \class{QIP}(k)=PSPACE for any $k\geq 3$. 

Concrete examples also exist where quantum reductions have proven
more powerful than classical counterparts in cryptography. Regev~\cite{Regev04} showed that the unique
shortest vector problem reduces to random subset sum problems via a
quantum reduction.  In addition, GapSVP and SIVP reduce to the
learning with errors (LWE) problem via quantum
reductions~\cite{Regev05}. There are no known classical reductions
between these problems under the same parameters. In fact, the only
worst-case to average-case reduction for the ring variant of LWE
(ring-LWE), which is a major competitor in the NIST standardization
effort of post-quantum cryptography, still relies on a quantum
reduction~\cite{LPR13}. Therefore, these are examples where quantum reductions
appear to be more powerful than classical reductions, when reducing
from worst-case to average-case problems.  Kawachi and
Yamakami~\cite{KY10} proved several hard-core predicates using quantum
reductions, inspired by earlier work on the quantum Goldreich-Levin
theorem~\cite{AC02} and the quantum algorithm for the Legendre
symbol~\cite{vDHI06}.

As existing negative results fail in the quantum setting and quantum reductions are shown to be more powerful than classical ones, it seems to be a live possibility that we can base cryptographic primitives on NP-hard or even QMA-hard problems via quantum reductions! 

\paragraph{Our results.} We give a general investigation of using
quantum reductions in basing the basic primitive one-way permutation
(function) on \np-hardness. Our contributions are summarized
below. 

We first generalize two central classical notions, \emph{locally
  random reduction} due to Feigenbaum and Lance Fortnow~\cite{FF93}
and \emph{worst-case to average-case reduction} due to Bogdanov and
Trevisan~\cite{BT06}, into the quantum setting.  {Locally random
  reductions} are commonplace in reductions between lattice problems
in~\cite{Ajtai96,Mic04,Mic07}, and Regev's reduction for
LWE~\cite{Regev05} naturally falls in our quantum analogue of locally
random reductions. Bogdanov and Trevision's notion is more general,
and our quantum formalization also allows to reason about more
powerful quantum reductions. We then give some characterizations of
the quantum reductions we define. We show that certain restrictions on
the quantum reductions will be too week to be interesting. For
instance, if a quantum reduction only issues \emph{entangled} queries
reminiscent of Bell states, it will not be able to base one-way
permutation on a language beyond \class{BQP}. On the other hand, we
give a couple of examples in the oracle setting that quantum
reductions are provably more powerful than their classical
counterparts.

Our main result is showing that the existence of some quantum
reductions implies unknown consequences in complexity.

\begin{theorem}\label{thm:intro}
  The existence of locally quantum reductions where the queries are
  non-adaptive and are according to smooth-computable distributions
  from an \class{NP}-complete problem (or \class{QMA-complete}
  problem) to the task of inverting one-way permutations implies
   \class{coNP} (or \class{coQMA}) $\subseteq$ \class{QIP($2$)}.
\end{theorem}
A distribution $\mathcal{D}$ is smooth-computable if its maximum and
minimum are only polynomially larger and smaller than the average, and
given any $x$, the probability $\Pr_{x\sim \mathcal{D}}[x]$ can be
computed efficiently. In particular, this rules out uniform
distributions, which are essential in existing worst-case to average-case
reductions~\cite{FF93,Ajtai96,Mic04,Mic07,Regev04_2,Ajtai97,Regev05}. We
also take an initial step towards the case of one-way
\emph{functions}. We show that given \emph{quantum-sampling oracle} of a one-way
function, any locally quantum reduction will result in the same negative consequence that \class{coNP} $\subseteq$ \class{QIP($2$)}.

As indicated before, direct translation of the classical proofs will
only give quantum interactive proof systems with more than three
messages, which will already coincide with \class{PSPACE}. We develop
a checking technique to mitigate the difficulty and design a
two-message protocol. Interestingly, our protocol illustrates a new
way of designing \class{QIP($2$)} than the few existing
ones~\cite{Rosgen09,Hayden14,kobayashi13}.

Admittedly, the containment  \class{coNP} $\subseteq$
QIP($2$) is not as strong as the classical result that the polynomial
hierarchy collapses. So far there is very little known about
\class{QIP($2$)} other than the simple fact that
\class{QMA}$\subseteq$ \class{QIP($2$)}$\subseteq$
\class{QIP($3$)}=\class{PSPACE}, and there are only a few
problems~\cite{Rosgen09,Hayden14,kobayashi13} proven to be in \class{QIP}($2$) that are not known to be in AM or QMA. It is an open question to pin
down where \class{QIP}($2$) stands precisely in the complexity zoo. As a partial
progress, we show in Section~\ref{sec:conp_qip2} an oracle problem
which is in \class{coNP}$^\mathcal{O}$ but not in  \class{QIP}($2$)$^\mathcal{O}$.

\paragraph{Overview of the proof of our main result.} In order to
describe the approach, we give more details about the classical
approach.  The classical proof strategy is to assume that a language
$L$ has a random reduction to another problem $L'$, and then construct
an interactive proof for $\overline{L}$ (the complement of $L$) that
induces a collapse.  For example, if $L$ is NP-complete and has a
random reduction to inverting a one-way permutation, then there is a
two round protocol for deciding if $x \in \overline{L}$.  The verifier
runs the generator $G$ to generate the queries for the one-way
permutation and sends them to the prover.  The prover then sends back
the answers.  Because the verifier can evaluate the one-way
permutation, the prover's answers can be checked, and then $R$ is run
to decide if $x\in L$.  Finally, the verifier can give the opposite
answer.  This results in a two round protocol for $\overline{L}$.
Therefore if such a reduction exists, then \class{coNP} $\subseteq$ AM, giving
a collapse.  There are much more complicated constructions when other
average-case problems are considered, for example, for reducing
worst-case NP-complete problems to distributional \class{NP} problems.
It is more difficult to find an interactive proof in this case because
if the prover answers $y\not\in L'$, then the verifier has no way to
verify this.  Nevertheless, with classical non-adaptive reductions it
is possible to construct a protocol for $\overline{L}$~\cite{BT06}.

Carrying over the proof strategy to the quantum setting requires
resolving several difficulties.  We use the unitaries $G$ and $R$ from
the reduction to construct a quantum interactive proof for
$\overline{L}$.  First the unitary $G$ is used to create superposition
queries which are sent to the prover.  An honest prover will answer
the superposition queries for the average-case language and send the
states back.  The verifier can then use the unitary $R$ to decide
whether to accept or reject.  The first difficultly that arises in
following this approach is that superposition queries are being used,
which makes it harder to verify that the prover is not cheating than
it is for classical answers.  Another limiting factor also immediately
arises in the quantum case that does not exist in the classical case.
For classical reductions and protocols, it is fine to create a
protocol with many (but still constant) rounds of communication,
because there is an equivalent two round protocol.  This is done
in~\cite{BT06} where upper and lower bound protocols are used to bound
the sizes of sets.  However, in the quantum case, we are limited to
finding quantum interactive proofs with only two rounds to begin with,
since \class{QIP}($m$)~=~\class{QIP}(3)~=~\class{PSPACE}.  Finding a quantum protocol that is
limited even to three rounds would only allow the conclusion that \class{coNP}
$\subseteq$ \class{QIP}($m$)~=~\class{PSPACE}, which does not yield a non-trivial
result.

The main technical challenge is to ensure that the prover provides the
answers honestly in superposition to the average-case problem.  A
cheating prover would try to return some other state that makes the
unitary $R$ answer in the opposite way than it should.  Should the
prover return such a state, the verifier must be able to detect this.
Our approach is to let the verifier create a superposition of two
states: the query state that is needed for the reduction, and a trap
state with the property that it can be used to detect that the prover
is cheating.  We show that there is a trap state so that whenever the
prover changes the query part of the superposition, then the trap part
of the superposition must also change, and that this can be detected
by the verifier.

\paragraph{Future directions.} There are several open questions.  Do
adaptive and/or non-smooth-computable quantum reductions from
\class{NP}-complete problems to inverting one-way permutations exist?
Can we generalize the result for permutations to other functions such
as 2-to-1 functions or even preimage-size-verifiable functions? We
observe that (Section~\ref{sec:fpsf}) if the oracle in the reductions
is capable of quantum sampling all solutions, then the existence of
such reductions from NP-hard problems to inverting k-to-1 surjective
functions also implies  \class{coNP} $\subseteq$
\class{QIP}($2$). Moreover, for cryptographic primitives which have more structures than one-way functions, it is possible that we can construct \class{QIP}($2$) protocol easily. Hence, we would like to know if we can also show that basing these cryptographic primitives on $\class{NP}$-hard problems via quantum reduction is unlikely. On the other hand, we are also interested to know if we can
find cryptographic primitives whose security are not likely to be
based on \class{NP}-complete problems under classical reductions, but
can be if quantum reductions are allowed. In the last, is it possible to rule out a
quantum reduction from NP-complete problems to average-case problems
in NP? Since we know very little about \class{QIP}($2$), can we show that
QIP($2$) $\neq$ QMA, \class{QIP}($2$) $\neq$ AM, \class{QIP}($2$) $\neq$\class{QIP}($3$), or
coNP$\nsubseteq$ \class{QIP}($2$)?

\section{Preliminaries}

For a finite set $X$, $\abs{X}$ denotes the size of $X$.  We use
$x\gets X$ to mean that $x$ is drawn uniformly at random from
$X$.  $\poly(\cdot)$ denotes an unspecified polynomial, and
$\negl(\secpar)$ denotes a negligible function in $\secpar$.  A
function $\epsilon(\secpar)$ is \emph{negligible} if for all
polynomials $p(\secpar)$, $\epsilon(\secpar) < 1/p(\secpar)$ for large
enough $\secpar$.  Classical efficient computation is described by
probabilistic polynomial time (PPT) algorithms.

We assume basic familiarity with quantum information formalism.  In
this paper, \emph{quantum register} represents a collection of qubits
that we view as a single unit.  We typically use capital letters to
denote a register and the Hilbert space associated with it.  A quantum
channel $\Phi$ describes any physically admissible transformation of
quantum states, which is mathematically a completely positive,
trace-preserving linear map.

We recall the definitions of \emph{quantum interactive proofs} (QIP)
and \emph{one-way permutations} (\owp).

\begin{definition}[\class{QIP($m$)}]
  A promise problem $A = (A_{yes},A_{no})$ is in the complexity class \class{QIP($m$)} if there exists a polynomial-time quantum
  verifier which exchanges at most $m$ quantum messages of length
  $O(\poly(|x|))$ with a prover and has the properties:
\begin{itemize}
\item (Completeness) For $x\in A_{yes}$, there exists a prover who can
  convince the verifier to accept $x$ with probability at least $2/3$.
\item (Soundness) For $x\in A_{no}$, no prover can convince the
  verifier with probability greater than $1/3$.
\end{itemize}
\end{definition}

Without loss of generality, the prover and the verifier can be
described as unitaries. It has been shown that
\class{QIP($m$)}=\class{QIP($3$)}=\class{PSPACE} for
$m\ge 3$~\cite{KW00,JJUW11}. It is known that completeness and
soundness can be reduced to negligibly small~\cite{JUW09}. In this
work, we focus on the class \class{QIP($2$)}.

\begin{definition}[One-way permutation]
\label{def:owp}
$f: \{0,1\}^*\rightarrow \{0,1\}^*$ is a  one-way permutation if
  \begin{itemize}
  \item for every $n$, $f$ is a polynomial-time computable permutation
    over $\{0,1\}^n$ by either quantum or classical algorithms, and
  \item for every quantum polynomial-time algorithm $A$,
    $\Pr_{x\gets \{0,1\}^n}(A(f(x)) = x) = \negl(\secpar)$.
  \end{itemize}   
\end{definition}

We denote inverting a one-way permutation as $\iowp$.

\subsection{Formal definitions of classical reductions}
\label{sec:clrr_w2a}

We review two defintions of worst-case to average-case reductions that
are central in the classical literature. We denote $\mathcal{P}'$ an
arbitrary decision, promise or search problem.  We will only consider
the case where $\mathcal{P}'$ corresponds to inverting one-way
permutations. We recall the basic notion of a \emph{distributional}
problem.

\begin{definition}[Distributional problem]
  \label{def:dp}
  Let $\mathcal{P}'$ be a problem and $\mathcal{D}$ a collection of
  distributions $\{\mathcal{D}_n\}_{n\in \mathbb{N}}$.  The distributional problem $(\mathcal{P}',\mathcal{D})$ is: given an
  instance $x$ chosen randomly according to $\mathcal{D}_n$, compute
  $\mathcal{P}'(x)$.
\end{definition}

One important notion of worst-to-average reductions is due to
Feigenbaum and Fortnow~\cite{FF93}, which we decribe here (with minor
adaption).

\begin{definition}[non-adaptive locally random reduction
  $(G,R)$]~\label{def:lrr} Let $k$ and $\ell$ be variables polynomial in the
  input length $n$, and $r$ is a random string chosen uniformly from
  $\{0,1\}^{\ell}$. A decision problem $\mathcal{P}$ is non-adaptively locally random reducible to a distributional
  problem $(\mathcal{P}',\mathcal{D})$ with error $\epsilon$ if there
  are polynomial-time algorithms $R$ and $G$ satisfying two conditions:
\begin{enumerate}
\item 
For $n\in \mathbb{N}$ and $x\in \{0,1\}^n$, it holds that $\mathcal{P}(x) = R(x,r, \mathcal{P}'(G(1,x,r)),\dots,\mathcal{P}'(G(k,x,r)) )$ for at least $1-\epsilon$ fraction of all $r\in\{0,1\}^{\ell}$.  
\item For $n\in \mathbb{N}$, $\{x_1,x_2\}\in \{0,1\}^n$, $1\leq i\leq k$, and $y$ a support of $\mathcal{D}$, it holds that $\Pr_r[G(i,x_1,r)=y] = \Pr_r[G(i,x_2,r)=y]=\Pr[y\sim \mathcal{D}_{n}]$.  
\end{enumerate}
\end{definition}

We note that Definition~\ref{def:lrr} is equivalent to the Definition
2.1 in~\cite{FF93} when $\epsilon=1/4$.  One can let each query be generated according to different fixed distributions in the locally random reduction. However, Definition~\ref{def:lrr} which assumes all queries are drawn from the same distribution is still general. This can be made without loss of generality since one can apply a random permutation to the queries before sending them to the oracle and undo the permutation before applying $R$.  This way, the
distributions of each query are the same.

If we only consider $\mathcal{D}$ to be the uniform distribution, then
this reduction is a special case considered by Feigenbaum et al.\
in~\cite{FKN90}.  One can also define \emph{adaptive} locally random
reductions by allowing the algorithm $G$ in Definition~\ref{def:lrr}
to generate queries depending on the previous queries and answers.

If $\mathcal{P} = \mathcal{P'}$, then the reduction is also called a
random-self reduction~\cite{FF93}.  It has been shown that the set of
complete problems in \class{PSPACE}, \class{EXP} and \class{\#P} are
random-self reducible~\cite{FF93}.  On the other hand, it has been
shown that \class{NP}-complete problems are not non-adaptive
random-self reducible unless the polynomial hierarchy collapses to the
third level~\cite{FF93}.

Note that the definition of Feigenbaum and Fortnow~\cite{FF93} assumes a perfect
solver for the average-case problem. This restriction is weakened in a
later work by Bogdanov and Trevisan~\cite{BT06}, which gives another
important notion of worst-to-average reductions below. To describe it,
we first define $\delta$-close problems.

\begin{definition}
  \label{def:close}
  A problem $\mathcal{P}''$ is {$\delta$-close} to another problem
  $\mathcal{P}'$ with respect to $\mathcal{D}$ if for all $n$, $\Pr_{x\sim \mathcal{D}_n} (\mathcal{P}''(x)\neq \mathcal{P}'(x)) < \delta$.
\end{definition}

\begin{definition}[non-adaptive worst-case to average-case reduction]
  \label{def:war} Let $k$ and $\ell$ be variables polynomial in the
  input length $n$, and $r$ is a random string chosen uniformly from
  $\{0,1\}^{\ell}$.  A decision problem $\mathcal{P}$ is non-adaptive worst-case to average-case reducible to
  $(\mathcal{P}',\mathcal{D})$ with average hardness $\delta$ and
  error $\epsilon$ if there are polynomial-time algorithms $R$ and $G$
  satisfying that:

 \begin{itemize}    
 \item[] for any $n\in \mathbb{N}$, and on all inputs
   $x\in \{0,1\}^n$, let $y_1,\dots,y_{k}$ be the outputs of $G(x,r)$.
   For any $\mathcal{P}''$ which is $\delta$-close to $\mathcal{P}'$
   with respect to $\mathcal{D}$,
      \[
        \Pr_r[R(x,r,\mathcal{P}''(y_1),\dots,\mathcal{P}''(y_{k}))=\mathcal{P}(x)]>1-\epsilon
        \, .\]
    \end{itemize}
\end{definition}
Similarly, an adaptive worst-case to average-case reduction is defined
by including previous queries and answers to the arguments of $G$.

It has been shown that \class{NP}-complete problems are not
non-adaptive worst-case to average-case reducible to themselves,
unless the polynomial hierarchy collapses to the third
level~\cite{BT06}.  In addition, the existence of a non-adaptive
worst-case to average-case reduction from \class{NP}-hard problem to
inverting a one-way function implies that the polynomial hierarchy collapses to the second
level~\cite{Akavia2006} and the existence of a worst-case to
average-case reduction from an \class{NP}-hard problem to inverting a
size verifiable one-way function implies that the polynomial hierarchy
collapses to the second level~\cite{Akavia2006,Bogdanov2015}.

\section{Formalizing quantum reductions}
\label{sec:qred_def}

In this section we define the quantum analogues of the two classical
notions of worst-to-average reductions (Definition~\ref{def:lrr}
and~\ref{def:war}). We then describe some characterizations and examples that illustrates diverse features of quantum reductions.

\begin{definition}[non-adaptive locally quantum
  reduction $(G,R)$]\label{def:qlrr}
  A decision problem $\mathcal{P}$ is non-adaptive locally
    quantum reducible to a distributional problem
  $(\mathcal{P}',\mathcal{D})$ with error $\epsilon$ by using $k$ queries if there are two polynomial-time
  implementable unitaries $R$ and $G$ such that for all $n$ and
  $x\in \{0,1\}^n$
  \begin{itemize}
  \item The generator $G$ creates $k$ superposition queries, with 
    query amplitudes based on the distribution $\mathcal{D}$: $G|0\rangle_{MV}|x\rangle =
      |Q_{x,1}\rangle\otimes\cdots\otimes|Q_{x,k}\rangle\ket{x} $
      where $|Q_{x,i}\rangle = \sum_{q\in{\mathbb{Z}_2^{m}}} \sqrt{d_q}|q,
      0\rangle_M|w_{x,i}(q)\rangle_V$ 
    for $i\in [k]$. Note that $d_q$ is the probability that $q$ is drawn from $\mathcal{D}_{n}$ and the register $V$ for the state $\ket{w_{x,i}}$ is the work register, which could be in arbitrary state.
  \item $R$ takes responses of the queries $|Q_{x,1}^H,\dots,Q_{x,k}^H\rangle$ and decides whether or not
    $\mathcal{P}(x)$ is true:
    \begin{eqnarray}
        R|Q_{x,1}^H,\dots,Q_{x,k}^H\rangle  =
      \sqrt{p}|\mathcal{P}(x)\rangle|\psi_{x}^0\rangle+
      \sqrt{1-p}|1-\mathcal{P}(x)\rangle|\psi_{x}^1\rangle \nonumber 
    \end{eqnarray}
      where $p\geq 1-\epsilon$ and $|Q_{x,i}^H\rangle = \sum_{q\in{\mathbb{Z}_2^{m}}}
      \sqrt{d_q}|q,\mathcal{P}'(q)\rangle_M|w_{x,i}(q)\rangle_V$, 
      for $i\in [k]$. 
  \end{itemize}
  \label{def:qff}
\end{definition}

\begin{definition}[non-adaptive quantum worst-case to average-case
  reduction]
  \label{def:qwar}
  A decision problem $\mathcal{P}$ is non-adaptive quantum
    worst-case to average-case reducible to
  $(\mathcal{P}',\mathcal{D})$ with average hardness $\delta$ and error $\epsilon$ if there
  are polynomial-time computable unitaries $R$ and $G$ such that for
  any $n\in \mathbb{N}$ and $x$
\begin{itemize}
    \item The generator $G$ creates $k$ superposition queries:
    \begin{eqnarray}
        G|0\rangle_{MV}|x\rangle =
      |Q_{x,1}\rangle\otimes\cdots\otimes|Q_{x,k}\rangle\ket{x}, 
      \label{eq:qbt_q}
    \end{eqnarray}
    where $|Q_{x,i}\rangle = \sum_{q\in{\mathbb{Z}_2^{m}}} c_{x,q,i}|q,
      0\rangle_M|w_{x,i}(q)\rangle_V, \quad \text{for } i\in [k]$. Note that the coefficients $c_{x,q,i}$ for $q\in{\mathbb{Z}_2^{m}}$ could be any complex numbers such that the sum of absolute squares is $1$. 
    \item $R$: for any $\mathcal{P}''$ which is $\delta$-close to
      $\mathcal{P}$ with respect to $\mathcal{D}$, 
    \begin{eqnarray}
        R|Q_{x,1}^H,\dots,Q_{x,k}^H\rangle  = \sqrt{p}|\mathcal{P}(x)\rangle|\psi_{x}^0\rangle+ \sqrt{1-p}|1-\mathcal{P}(x)\rangle|\psi_{x}^1\rangle,
        \label{eq:qbt_a}
    \end{eqnarray}
where $p\geq 1-\epsilon$ and $|Q_{x,i}^H\rangle = \sum_{q\in{\mathbb{Z}_2^{m}}}
      c_{x,q,i}|q,\mathcal{P}''(q)\rangle_M|w_{x,i}(q)\rangle_V$, for $i\in [k]$. 
\end{itemize}
The variables $m$ and $k$ are polynomial in the input length $n$.   
\end{definition}

Compared to locally quantum reductions, quantum worst-case to
average-case reductions do not require the queries to be drawn from a
certain distribution.  Instead, we consider an oracle for
$\mathcal{P}'$ that can err sometimes, which is captured by
$\delta$-close problems $\mathcal{P}''$. $1-p$ is called the error of the reduction.  The choice of $p=2/3$ is
arbitrary, since it can be reduced effectively.

%----------------------------------------%
\subsection{Separation examples}
\label{sec:sepex}
%----------------------------------------%

We give two examples demonstrating the distinct landscapes of
classical and quantum worst-case to average-case reductions.  Namely,
relative to an oracle and under reasonable computational assumptions,
there exist a worst-case problem and an average-case problem such that
no classical reductions exist whereas they admit an efficient quantum
reduction. In fact, the quantum reduction issues non-adaptive
\emph{classical} queries only. This makes the separation examples
strong.

The idea behind the examples is simple.  We design the average-case
problem in such a way that to make a meaningful query to a solver for
this average-case problem, one has to solve a problem that is (assumed
to be) hard for classical algorithms but easy on a quantum computer.
Our first example is based on a oracle problem provably hard
classically (Simon's Problem), and the quantum reduction needs quantum
access to the oracle.  The second example needs to assume the
existence of problem in BQP that is outside BPP (e.g.,
factorization). However we remove the need of quantum access to the
oracle as a reward. Both constructions rely on the following
assumption.

\begin{asn} There exists language $L\notin$ \class{BQP} (hence
  $L\notin$ \class{BPP} too) that admits a random self-reduction
  $L\red{(R,G)} (L,D)$ for some distribution $D$.
\label{asn:tqbf}
\end{asn}
A candidate is the \class{PSPACE}-complete problem True Quantified Boolean Formula (\lang{TQBF}), which
is known to have a \emph{non-adaptive} random
self-reduction~\cite{FF93}.  Assumption~\ref{asn:tqbf} will follow, if
\class{BQP} $\subsetneq$ \class{PSPACE}.  Hereafter we treat $G$
as non-adaptive in Assumption~\ref{asn:tqbf} for simplicity.

Let $N = 2^n$, and for each $i \in [N]$, let
$f_i: \bool^n \to \bool^n$ be some function, and $s_i\in \bool^n$.  We
define an oracle $O: = O_{s_0, \ldots, s_{N-1}}$ that generalizes
Simon's oracle~\cite{Simon97}.
\begin{eqnarray*}
  O: \ket{i,x,y,z} \to \ket{i,x, y\oplus f_i(x), z} \, ; \quad \text{ where } f_i(x) = f_i(x') \text{ iff.  } x' = x\oplus s_i \, .
\end{eqnarray*}
We assume that all $s_i, i\in [N]$ are chose uniformly at random.  As
an immediate corollary of Simon's result.  We have that

\begin{lemma}
  Given $O$, any classical algorithms needs $\Omega(2^{n/2})$ queries
  to $O$ to find $s_i$ for some $i \in [N]$.  For any $i \in [N]$,
  there is a quantum algorithm that can find $s_i$ with $O(n^2)$
  queries and time.
\label{lemma:gsimon}  
\end{lemma}

\mypar{Construction 1} We construct our first separation example.  
\begin{itemize}
\item $L_1 = \lang{TQBF} = \{\phi = \phi(v_1,\ldots, v_n)\}$
  containing satisfiable quantified $n$-variable formulae in
  3-CNF.  Let $L_1^O$ be the language $L_1$ relative to oracle $O$,
  which simply ignores $O$.
\item $\hat L_1^O: = \{x=(i,s,\phi): \text{$s = s_i$ and $\phi$ is
    true}\}$.  We associate $\hat L_1^O$ a distribution $\hat D_1$,
  which is uniform on $N\times \bool^n$ and samples a formula
  according to $D$ (the distribution in Assumption~\ref{asn:tqbf}).
\end{itemize}

\begin{theorem}
  Under Assumption~\ref{asn:tqbf}, there does not exist a PPT
  reduction from $L_1^O$ to $(\hat L_1^O, \hat D_1)$.  In contrast,
  there is a quantum poly-time \emph{non-adaptive} reduction
  $L_1^O \red{R_1^Q,G_1^Q} (\hat L_1^O,\hat D_1)$.
  \label{thm:sepsimon}
\end{theorem}

\begin{proof}
  Let $A$ be an algorithm that solves the average-case problem
  $(\hat L_1^O, \hat D_1)$.  For simplicity, we assume that $A$ is a
  perfect decider, i.e., for a random input
  $x = (i,s,\phi)\gets \hat D_1$, $A(i,s,\phi) = 1$ iff.  $s= s_i$ and
  $\phi = 1$.  Any classical reduction is unable to find $s_i$ in
  polynomial time, hence the solver $A$ is useless.  Formally speaking,
  if there were such a reduction $L_1^O \red{}(\hat L_1^O,\hat D_1)$,
  one can turn it into an efficient solver for Simon's problem or an
  efficient decider for $L$.  This violates Lemma~\ref{lemma:gsimon} or
  Assumption~\ref{asn:tqbf}.

  For the second part, we construct a quantum reduction
  $(R^Q_1,G^Q_1)$ as follows.  Recall that there is a random
  self-reduction $L \red{(R,G)} (L,D)$.  Given a worst-case input
  $\phi$, $G_1^Q(\phi)$ runs $G(\phi)$ to get random
  $\{\phi_j\}_{j=1}^k$.  Then for $j=1,\ldots k$, $G_1^Q$ generates
  random $i_j \gets [N]$, and runs Simon's algorithm to find $s_{i_j}$
  efficiently.  Then the queries to $A$ are
  $\{i_j,s_{i_j}, \phi_j\}_{j=1}^k$, which are correctly distributed
  according to $\hat D_1$.  Therefore $A$ will respond correctly with
  $\{\phi_j \overset{?}{=}1\}$.  Then $R_1^Q$ runs the decision
  procedure $R$, which correctly decides $\phi$.
\end{proof}

\begin{remark} We have designed $O$ to encode exponentially many
  instances of Simon's problem for the technicality of
  \emph{non-uniform} reductions.  Because otherwise, a classical
  reduction could hardwire the solution $s$ and make use of an
  average-case solver.
\end{remark}

\mypar{Construction 2} We give another separation example.  It is still
in the oracle setting, and we need to make an additional
assumption.  What we gain is that the quantum reduction does not need
quantum access to the oracle, as opposed to Example 1 where we need to
run Simon's algorithm with quantum access to the oracle.

\begin{asn}
  There exists a classically secure one-way function $f: X\to Y$,
  which is invertible by an efficient quantum algorithm.
\label{asn:owf}  
\end{asn}

A natural candidate would be adaption of \lang{Factorization}.  Let
$(p,q)\gets \Gen(1^n)$ be an efficient algorithm that generates two
large primes at random, and define $f(p,q) = pq$.  Then it is
reasonable to assume that there exists a $\Gen$ algorithm relative to
which $f$ is hard to invert.  In fact this is necessary for the
\class{RSA} assumption, which is the basis of modern public-key
cryptography.  This assumption is hence likely to be true given the
current state of art.

Given a function $f$ as in Assumption~\ref{asn:owf}, we define an
oracle $H: i \mapsto y_i$ for $i \in [N]$.  Here we sample
$z_i \gets X$ randomly and set $y_i : = f(z_i)$.

\begin{itemize}
\item $L_2 = \lang{TQBF} = \{\phi = \phi(v_1,\ldots, v_n)\}$
  containing satisfiable quantified $n$-variable formulae in
  3-CNF.  Let $L_2^O$ be the language $L$ relative to oracle $H$, which
  ignores $H$.
\item $\hat L_2^H: = \{x=(i,z,\phi): \text{$f(z) = y_i$ and $\phi$ is
    true}\}$.  We associate $\hat L_2^H$ a distribution $\hat D_2$,
  which is uniform on $[N] \times X$ and samples a formula according
  to $D$ (the distribution in Assumption~\ref{asn:tqbf}).
\end{itemize}

\begin{remark} For the same reason as above, we introduce the oracle
  $H$ to encode superpolynomial-many instances of inverting $f$ to
  avoid a non-uniform classical reduction that can hardwire solutions
  to (at most poly-many) inversion instances.
\end{remark}

Following similar arguments to Theorem~\ref{thm:sepsimon}, we can
prove the theorem below.  The only change is that in the quantum
reduction, we query \emph{classically} a random index $i_j$ to $H$,
and obtain $y_{i_j}$.  Then we run Shor's
algorithm to find $z_{i_j} : = f^{-1}(y_{i_j})$, and then form correct
queries to the solver of $(\hat L_2^H, \hat D_2)$.

\begin{theorem}
  Under Assumption~\ref{asn:tqbf} and~\ref{asn:owf}, there does not
  exist a classical reduction from $L_2^H$ to $(\hat L_2^H, \hat D_2)$.  In
  contrast, there is a quantum poly-time \emph{non-adaptive} reduction
  $L_2^H \red{R_2^Q,G_2^Q} (\hat L_2^H,\hat D_2)$.
  \label{thm:sepfact}  
\end{theorem}

%------------------------------------------
\subsection{Discussion of the definitions: special cases}
\label{sec:special}
%----------------------------------------%

\subsubsection{Classical queries.} If $G$ outputs classical queries, then
we can immediately derive a negative result, analogous to a classical
result by Brassard~\cite{Brassard79}. In the following,
$L<_{R,G} \cal{P}'$ denotes a non-adaptive quantum reduction from $L$
to $\cal{P}'$, where $G$ generates oracle queries and $R$ decides $L$
according to the oracle's responses. We will add more constraints to
the reductions accordingly in this section.

\begin{theorem}~\label{thm:qrcq} If there is a non-adaptive quantum
  reduction $L \red{R,G} \iowp$ where $G$ only issues classical queries, then $L \in$ \class{QIP$(2)$} with classical interactions.
 \label{thm:w2owpc}  
\end{theorem}

\begin{proof} 
  The protocol is as follows: The verifier first applies $G$ to
  generate queries and sends these queries to the prover.  Then, the
  prover simulates the oracle for $\iowp$ and sends the responses
  back.  Finally, the verifier checks if the responses are correct by
  computing the permutation.  If the prover is not cheating, the
  verifier applies $R$ and accepts if the reduction accepts.  Otherwise, the
  verifier rejects.  Note that the prover can only give the correct
  answer for $\iowp$.  Otherwise, the verifier rejects.
\end{proof}

\subsubsection{EPR queries.} As another special case, we consider that
$w_{x,i}(q)$ is the identity function in Definition~\ref{def:qff}.
Namely $G$ generates $k$ identical copies of
$\ket{\Psi}^{\otimes m} = \sum_{q} \ket{q}\ket{q}$, where
$\Psi = \frac{1}{\sqrt{2}} (\ket{00} + \ket{11})$ is an {EPR}-pair.
Then half of the EPR-pairs are submitted as queries to the solver of
the average-case problem.  Note the reduced density of each query is
totally mixed, and this looks a natural generalization of classical
uniform queries.  Nonetheless, we show that this is too strong a
constraint that trivializes the study of worst-to-average reductions,
as far as \class{OWP} is concerned.

\begin{proposition} If there is a reduction $L \red{R,G} \iowp$ where $G$
  issues EPR-queries, then $L \in$ \class{BQP}.
 \label{prop:w2owpepr}  
\end{proposition}

\begin{proof} Observe that a uniform superposition over a set is
  \emph{invariant} under an arbitrary permutation.  This means that a
  quantum reduction could create the correct state that otherwise
  would require invoking an inverting oracle $I$ of the \class{OWP}
  $f$.  Namely applying $I$ on an EPR query gives us
  \begin{eqnarray*}
    \sum_q\ket{q,q} \overset{I}{\mapsto} \sum_{q}
    \ket{q,q,f^{-1}(q)} \, .    
  \end{eqnarray*}
  This can be created without help of $I$ as follows:
  \begin{align*}
    \ket{0,0,0} \mapsto \sum_q\ket{q,0,0} & \overset{f}{\mapsto} \sum_{q}
      \ket{q,f(q),f(q)} = \sum_{q'} \ket{f^{-1}(q'),q',q'}\\
    &\overset{\mathrm{SWAP}_{1,3}}\mapsto \sum_{q'} \ket{q',q',f^{-1}(q')}\, .   \end{align*}
\end{proof}

\begin{remark} Consequently, it is \emph{necessary} that the query
  states maintain more sophisticated correlations between the query
  register and the work register of the reduction.  We note the same
  phenomenon also occur classically.  Namely, although the marginal
  distribution of each query is \emph{uniformly} random, it is
  important that the internal state of the reduction should not be
  independent of the queries.  Otherwise, existence of such a reduction
  will trivialize the language under consideration to fall in
  \class{BPP}.
\end{remark}

%----------------------------------------%
\section{On uniform locally quantum reductions}
\label{sec:ulqr}
%----------------------------------------%
We prove our main result in this section. We first consider
quantum reductions that make \emph{one-query} only.  It demonstrates
the main idea of our general result with a cleaner analysis.
Section~\ref{sec:ff_reduction} will handle multiple non-adaptive
queries.
%----------------------------------------%
\subsection{Uniform one-query locally quantum reductions}
\label{sec:one_reduction}
%----------------------------------------%

Let $f$ be a one-way permutation on $\bool^n$, and let
$U_f: \ket{x,y} \mapsto \ket{x,y\oplus f(x)}$ be a unitary quantum
circuit computing it. Note that $U_f\ket{f^{-1}(x),0} = \ket{f^{-1}(x),x}$.  A uniform one-query locally quantum reduction for $L$ works as follows:

\begin{eqnarray}
  \ket{x, 0}\ket{0} &\stackrel{\generator}{\longrightarrow}&
                                                      \frac{1}{\sqrt{2^{m}}} \sum_{q\in\mathbb{Z}_2^m} \ket{q,0}\ket{w_{x}(q)}\label{eq:o_qlrr_1}\\
             &\overset{O_{f^{-1}}}{\longrightarrow}& \frac{1}{\sqrt{2^{m}}} \sum_{q\in \mathbb{Z}_2^m} \ket{q,f^{-1}(q)}\ket{w_{x}(q)}
             \stackrel{\reduction}{\longrightarrow}
             a_0|0\rangle|\psi_{x,0}\rangle+a_1|1\rangle|\psi_{x,1}\rangle,\label{eq:o_qlrr_2}
\end{eqnarray}
where $|a_1|^2\geq 1-\epsilon$ if $x\in L$ and $|a_1|^2\leq \epsilon$
if $x\notin L$ is the accepting probability of the reduction.

\begin{theorem}~\label{thm:one_main} Suppose there exists a one-query
  uniform locally quantum reduction with exponentially small error
  $\epsilon$ from a worst-case decision problem $L$ to the task of
  inverting a polynomial-time computable permutation.  Then there
  exists a $QIP($2$)$ protocol with completeness $1-\epsilon/2$ and
  soundness $1/2+2\sqrt{\epsilon}$ for $\overline{L}$
\end{theorem}

\subsubsection{The protocol for $\overline{L}$.}
We are given the uniform one-query locally quantum reduction $(G,R)$.
We enlarge the size of register $V$ and define a unitary $\cnot$ which
performs a CNOT on the first register of $M$ into the second register
of $V$:
$\ket{q,x}_M\ket{y,z}_V \stackrel{\cnot}{\longrightarrow}
\ket{q,x}_M\ket{y,z\oplus q}_V$.  The whole protocol takes place in
the space
$\mathcal{H}_P\otimes\mathcal{H}_M\otimes\mathcal{H}_V\otimes\mathcal{H}_{\Pi}$
where $P$ is the private register of the prover, $M$ is the register
exchanged between the prover and the verifier, $V$ and $\Pi$ are
registers which are private to the verifier.

We describe some states that are crucial in the protocol.

\begin{itemize}
\item The verifier prepares the state $|S\rangle_{MV\Pi} = \frac{1}{\sqrt{2}}(|Q\rangle_{MV}|0\rangle_{\Pi}+|T\rangle_{MV}|1\rangle_{\Pi})$, where
    \begin{eqnarray}
    |Q\rangle_{MV} = \frac{1}{\sqrt{2^n}}\sum_{q\in \mathbb{Z}_2^m} |q,0\rangle_M|w_x(q),q\rangle_V\label{eq:Q}
    \end{eqnarray}
    without the extra copy of $q$ in the register $V$ is the query state generated from $G$ as in Equation~\ref{eq:o_qlrr_1}, and
    \begin{eqnarray}
        |T\rangle_{MV} = \frac{1}{\sqrt{2^n}}\sum_{q\in \mathbb{Z}_2^m} |q,0\rangle_M|0,q\rangle_V\label{eq:T}
    \end{eqnarray}
    is the trap state, which will be used to catch a cheating prover.  
    \item The honest prover replies $|S^H\rangle_{MV\Pi}=\frac{1}{\sqrt{2}}(|Q^H\rangle_{MV}|0\rangle_{\Pi}+|T^H\rangle_{MV}|1\rangle_{\Pi})$, where
    \begin{eqnarray}
        |Q^H\rangle_{MV} &=& \frac{1}{\sqrt{2^n}}\sum_{q\in \mathbb{Z}_2^m}
        |q,f^{-1}(q)\rangle_M|w_x(q),q\rangle_V\label{eq:q_response}\\
      |T^H\rangle_{MV} &=& \frac{1}{\sqrt{2^n}}\sum_{q\in \mathbb{Z}_2^m}
      |q,f^{-1}(q)\rangle_M|0,q\rangle_V.\label{eq:t_response}
    \end{eqnarray}
    The state $|Q^H\rangle$ without the extra copy of $q$ in register $V$
    is the state the actual reduction $R$ gets after querying the
    oracle as in Equation~\ref{eq:o_qlrr_2}.  The state $|T^H\rangle$ can be
    mapped to $|0\rangle_{MV}$ efficiently as shown below.  This gives
    the verifier an efficient way to check if $|T^H\rangle$ is changed significantly.
\end{itemize}

We do the following to map $|T^H\rangle_{MV}$ back to $|0\rangle_{MV}$
efficiently
\begin{eqnarray*}
  \sum_{q\in\mathbb{Z}_2^m}|q,f^{-1}(q)\rangle_M|0,q\rangle_V &\xrightarrow{C}& 
                                                                                \sum_{q}|q,f^{-1}(q)\rangle_M|0,q\oplus
                                                                                q\rangle_V\\
                                                              &\xrightarrow{U_{f}}&\sum_{q}|q\oplus
                                                                                    f(f^{-1}(q)),f^{-1}(q)\rangle_M|0,0\rangle_V
                                                              \xrightarrow{F} |0,0\rangle_{M}|0,0\rangle_V\label{eq:T_0_3}.
\end{eqnarray*}
Here $U_f$ is applied from the second register of $M$ into the first,
and $f$ is applied to the second register of $M$.  The last two steps use
the property that $f$ can be evaluated efficiently and $f$ is a permutation.

Given the reduction $(G,R)$, we can get a \class{QIP($2$)} protocol
for $L$ by answering the same as $R$ and a protocol for $\overline{L}$
by flipping $R$'s answer.  The \class{QIP}($2$) protocol for $\overline{L}$ is
described in Protocol~\ref{fig:protocol1query}.

\floatname{algorithm}{Protocol}
\begin{algorithm}[h]
  \begin{mdframed}[style=figstyle,innerleftmargin=10pt,innerrightmargin=10pt]
The protocol takes place in the space
      $\mathcal{H}_P \otimes \mathcal{H}_{M} \otimes \mathcal{H}_V
      \otimes \mathcal{H}_{\Pi}$ 
      where $P$ is the private register of the prover, $M$ is the
      register exchanged between the prover and the verifier, and
      $V$ and $\Pi$ are registers which are private to the verifier.
    \begin{enumerate}
    \item \emph{The verifier's query.}  The verifier
      prepares $\ket{S}_{MV\Pi} := \frac{1}{\sqrt{2}}(|Q\rangle_{MV}|0\rangle_{\Pi}+|T\rangle_{MV}|1\rangle_{\Pi}).$
    
    The message register $M$ is sent to the prover, and the verifier keeps $V$ and $\Pi$.
    This is generated by conditioning on the register $\Pi$,
    which is initialized in $\ket{+}$.  If $\Pi = 0$, $G$ is applied and then
    $q$ is copied to the second part of the verifier's internal register $V$,
    which produces $\ket{Q}_{MV}$.  If $\Pi = 1$, compute the
    Fourier transform followed by CNOT to create $\ket{T}_{MV}$, a trap
    state we use to catch a cheating prover.
    
  \item \emph{The prover's response.}  The prover applies some unitary $U_{PM}$ on register $M$ and its private register $P$ and sends the message register back to the verifier.  
  \item \emph{The verifier's verification.} The verifier applies
    $C$ to erase $q$ in $V$.  The verifier then measures $\Pi$ to obtain
    $b \in \{0,1\}$, and does the following:
    \begin{itemize}
    \item \text{(Computation verification)} If $b=0$, apply $R$ on $MV$
      and measure the output qubit.  Accept if the outcome is $0$.
    \item \text{(Trap verification)} If $b=1$, apply $V_T$ on $MV$
      and measure $MV$.  Accept if the outcome is all $0$, (i.e., if the reduction rejects).  
    \end{itemize}
  \end{enumerate}
  \caption{\class{QIP}($2$) protocol for $\overline{L}$ using a one-query locally quantum reduction.}
  \label{fig:protocol1query}
  \end{mdframed}
\end{algorithm}

\subsubsection{Proof of
  Theorem~\ref{thm:one_main}.}~\label{sec:proof_one_main} 
We first prove a few useful lemmas.  Lemma~\ref{lemma:compm} is an
immediate consequence of the fact two purifications of the same state
are related by an isometry.  In fact, this exactly explains why the
entanglement fidelity is well defined~\cite{SCh96}.

\begin{lemma} \label{lemma:compm}
Let $\rho_A$ be a state in some Hilbert space.  Let
  $\ket{\phi}_{AB}$ and $\ket{\psi}_{AB}$ be two purifications of
  $\rho_A$, i.e.,
  $\tr_B(\kera{\phi}_{AB})=\tr_B(\kera{\psi}_{AB}) = \rho_A$.  Let
  $\Psi_A: \mathcal{H}_A\rightarrow \mathcal{H}_A$ be a quantum channel.  
  Let
  $\rho_{AB}: = (\Psi_A\otimes I_B)(\kera{\phi}_{AB})$ and
  $\sigma_{AB}:= (\Psi_A\otimes I_B)(\kera{\psi}_{AB})$, where
  the notation $\Psi_A\otimes I_B$ means that the channel is only applied on the space $\mathcal{H}_A$ and space $\mathcal{H}_B$ is not changed.  
  Then
  $\langle \phi| \rho_{AB} |\phi\rangle = \langle \psi | \sigma_{AB}
  |\psi\rangle$.
\end{lemma}
\begin{proof}  
Observe that (e.g., by Schmidt decomposition) there is
  a unitary $U_B$ operating only on $B$ such that
  $I_A\otimes U_B \ket{\psi}_{AB} = \ket{\phi}_{AB}$.  Then
  \begin{eqnarray}
    &&\bra{\phi} \rho_{AB}\ket{\phi}\nonumber=\bra{\phi} (\Psi_A\otimes I_B) ((\ket{\phi}\bra{\phi})_{AB}) \ket{\phi}\nonumber\\
    &=& \bra{\psi}(I_A\otimes U_B^{\dagger}) (\Psi_A\otimes I_B)(I_A\otimes U_B(\ket{\psi}\bra{\psi})_{AB}I_A\otimes U_B^{\dagger})
        (I_A\otimes U_B)\ket{\psi}\nonumber \\
    &=& \sum_{\ell} \bra{\psi}(I_A\otimes U_B^{\dagger}) (E_A^{\ell}\otimes I_B)(I_A\otimes U_B(\ket{\psi}\bra{\psi})_{AB}I_A\otimes U_B^{\dagger})\nonumber\\
    &&(E_A^{\ell\dagger}\otimes I_B)
        (I_A\otimes U_B)\ket{\psi}    \label{eq:ops}\\ 
    &=& \sum_{\ell} \bra{\psi}(E_A^{\ell}\otimes I_B)((\ket{\psi}\bra{\psi})_{AB})(E_A^{\ell\dagger}\otimes I_B)
        \ket{\psi}    \label{eq:ops2}\\
    &=& \bra{\psi} (\Psi_A\otimes I_B)((\ket{\psi}\bra{\psi})_{AB}) \ket{\psi}= \bra{\psi} \sigma_{AB}\ket{\psi}.\nonumber
  \end{eqnarray}  
The operators $\{E_A^{\ell}\}$ in Equation~\ref{eq:ops} are the operation elements of the channel $\Psi_{A}$, where $(\Psi_{A}\otimes  I_B)((\ket{\phi}\bra{\phi})_{AB}) = \sum_{\ell} (E_{A}^{\ell}\otimes I_B)((\ket{\phi}\bra{\phi})_{AB})(E_{A}^{\ell\dagger}\otimes I_B)$ and $(\Psi_{A}\otimes  I_B)((\ket{\psi}\bra{\psi})_{AB}) = \sum_{\ell} (E_{A}^{\ell}\otimes I_B)((\ket{\psi}\bra{\psi})_{AB})(E_{A}^{\ell\dagger}\otimes I_B)$.   Equation~\ref{eq:ops2} is correct due to the property that $(A\otimes B)(C\otimes D)=(AC)\otimes(BD)$.  
\end{proof}

Without loss of generality, we
can always represent the prover's operator $U_{PM}$ as $U'_{PM} O_{f^{-1}}$ where $U'_{PM}$ is an arbitrary unitary
the cheating prover may apply.
Let 
\begin{align*}
&\sigma^{U'}_Q=\Tr_P(U'_{PM}\otimes I_V(|0\rangle\langle 0|\otimes |Q^H\rangle\langle Q^H|)U'^{\dagger}_{PM}\otimes I_V)\mbox{, and}\\    
&\sigma^{U'}_T=\Tr_P(U'_{PM}\otimes I_V(| 0\rangle\langle  0|\otimes |T^H\rangle\langle T^H|)U'^{\dagger}_{PM}\otimes I_V). 
\end{align*}
The following claim shows that for any unitaries the prover applies, the change on $|T^H\rangle$ is as much as the change on $|Q^H\rangle$.

\begin{lemma}\label{claim:u_indistinguishable}
For an arbitrary $U'_{PM}$, let $|Q^H\rangle_{MV}$ and $|T^{H}\rangle_{MV}$ be as defined in Equation~\ref{eq:q_response} and Equation~\ref{eq:t_response}, and let $\sigma^{U'}_Q$ and $\sigma^{U'}_T$ be as above.  Then $\langle Q^H|\sigma_Q^{U'}|Q^H\rangle = \langle T^H|\sigma_T^{U'}|T^H\rangle$.  
\end{lemma}
\begin{proof}
  We represent the prover's behavior $U'_{PM}$ on the state $\ket{Q^H}$ and $\ket{T^H}$ as a noisy channel $\Psi_M^{U'}$ operating on register $M$, which is formally defined as follows: 
  \begin{eqnarray}
    \mbox{For all }\rho\in \mathcal{H}_P\otimes\mathcal{H}_M\otimes\mathcal{H}_V,\mbox{ }(\Psi_M^{U'}\otimes I_V)(\rho) := \Tr_P((U'_{PM}\otimes I_V)\rho(U'^{\dagger}_{PM}\otimes I_V)).  \nonumber
  \end{eqnarray}
  Therefore, 
  \begin{align*}
    &(\Psi_{M}^{U'} \otimes I_V)(\ket{Q^H}\bra{Q^H}) = \sigma^{U'}_Q
    &(\Psi_{M}^{U'}\otimes I_V)(\ket{T^H}\bra{T^H}) = \sigma^{U'}_T.  
  \end{align*}
  $|T^H\rangle$ and $|Q^H\rangle$ are actually two purifications of a
  mixed state on register $M$ since $\Tr_V(\ket{Q^H}\bra{Q^H})=
  \Tr_V(\ket{T^H}\bra{T^H})$, and by Lemma~\ref{lemma:compm}, we can conclude that 
  \[
   \langle Q^H|\sigma_Q^{U'}|Q^H\rangle = \langle T^H|\sigma_T^{U'}|T^H\rangle.  
   \]
\end{proof}

Given a state $\ket{\phi}$ and a projector $\Pi_S$, Lemma~\ref{lemma:maxproj} shows that the state $\rho$ which maximizes the quantity $\Tr(\Pi_S \rho) + \bra{\phi} \rho \ket{\phi}$ is the bisector of $\ket{\phi}$ and its projection on $\Pi_S$.  
\begin{lemma}\label{lemma:maxproj}
Let $S \subseteq \hilbert$ be a subspace and $\Pi_S$ be
  the projection operator on $S$.  Let $\ket{\phi}$ be a state such
  that $\bra{\phi} \Pi_S \ket{\phi} = \sin^2 \theta$, for some
  $\theta \in [0,\pi/2]$.  Then for any density operator
  $\rho \in D(\hilbert)$,
  $\Tr(\Pi_S \rho) + \bra{\phi} \rho \ket{\phi} \leq 1 + \sin \theta$.
\end{lemma}

\begin{proof} We first prove this lemma for any pure state
  $\rho = \kera{\psi}$.  Let $\text{dim}(S) = k$.  Let
  $\ket{v_0} := \frac{\Pi_S \ket{\phi}}{\|\Pi_S\ket{\phi}\|}$ and
  $\ket{v_k}: = \frac{\ket{\phi} - \ket{v_0}}{\|\ket{\phi} -
    \ket{v_0}\|}$.  Clearly, $\ket{v_k} \perp \ket{v_0}$, and
  $\ket{\phi} = \sin\theta \ket{v_0} + \cos\theta \ket{v_k}$.  Then we
  pick $\{\ket{v_1}, \ldots, \ket{v_{k-1}}\}$ in $S$ such that
  $\{\ket{v_0}, \ldots, \ket{v_{k-1}}\}$ form an orthonormal basis for
  $S$.  As a result, $\{\ket{v_0}, \ldots, \ket{v_k}\}$ will be an
  orthonormal basis for $\tilde S : = \text{span}(S\cup
  \ket{\phi})$.  Consider any $\rho = \kera{\psi}$ with
  $\ket{\psi} \in \tilde S$.  Then $\ket{\psi}$ can be written as 
  \begin{equation}\label{eq:on_the_plate}
    \ket{\psi} = \sum_{i=0}^k \alpha_i \ket{v_i}, \quad \sum_i
    |\alpha_i|^2 = 1 \, .
  \end{equation}
  We have that
  \begin{eqnarray*}
    \bra{\phi} (\kera{\psi}) \ket{\phi} &=& |\alpha_0\sin\theta  +
                                            \alpha_k\cos\theta|^2 =
                                            |\alpha_0|^2 \sin^2\theta
                                            + |\alpha_k|^2\cos^2\theta\\
                                            &&+ \sin \theta\cos\theta(\alpha_0
                                            \alpha^*_k + \alpha_0^* \alpha_k) \,; \\
    \Tr(\Pi_S \kera{\psi}) &=& \sum_{i=0}^{k-1} |\alpha_i|^2 = 1 -
                               |\alpha_k|^2\, .  
  \end{eqnarray*}
  Therefore
  \begin{eqnarray}
    && \Tr(\Pi_S \kera{\psi})  +     \bra{\phi} (\kera{\psi})
       \ket{\phi} \label{eq:bound_1}\\
    &=& 1 + \sin^2\theta |\alpha_0|^2 + (\cos^2\theta - 1)
        |\alpha_k|^2 + \sin\theta\cos\theta (\alpha_0\alpha_k^* + \alpha_0^*\alpha_k)\nonumber\\
    &=& 1 + \sin\theta\cdot\left(\sin\theta(|\alpha_0|^2-|\alpha_k|^2) +
        \cos\theta(\alpha_0\alpha_k^* + \alpha_0^*\alpha_k) \right) \nonumber\\
    &\leq & 1 + \sin\theta\cdot(\sin\theta (|\alpha_0|^2-|\alpha_k|^2) + 2\cos\theta(|\alpha_0||\alpha_k|)).\label{eq:bound_2}
  \end{eqnarray}
  Since the expression in Equation~\ref{eq:bound_2} is strictly
  increasing with $|\alpha_0|$ and independent to $|\alpha_1|,\dots |\alpha_{k-1}|$, we can suppose  the optimal $|\psi\rangle$
  for Equation~\ref{eq:bound_1} is on the subspace spanned by
  $|v_0\rangle$ and $|v_k\rangle$ without loss of generality.  Thus we let $|\alpha_0| = \cos\theta_0$ and $|\alpha_k| = \sin\theta_0$ and the upper bound for Equation~\ref{eq:bound_2} as below
  \begin{eqnarray*}
  &&1 + \sin\theta\cdot(\sin\theta (|\alpha_0|^2-|\alpha_k|^2) + 2\cos\theta(|\alpha_0||\alpha_k|)) \\
  &=& 
  1+\sin\theta(\sin\theta(\cos^2\theta_0 -\sin^2\theta_0)+ 2\cos\theta\cos\theta_0\sin\theta_0)\\
  &=& 1+\sin\theta(\sin\theta \cos2\theta_0 + \cos\theta\sin2\theta_0)\\
  &=& 1+\sin\theta(\sin (\theta + 2\theta_0))\leq 1+\sin\theta.  
  \end{eqnarray*}
  The maximum is achieved when
  $\theta_0= \frac 1 2( \pi/2 - \theta)$, i.e., when
  $\ket{\psi}$ bisects $\ket{\phi}$ and $\ket{v_0}$.  
   
  For an arbitrary mixed state
  $\rho := \sum_{i} p_i \kera{\psi_i}$ with $\sum_i p_i =1$,
  $p_i \geq 0$.
  \begin{equation*}
    \Tr(\Pi_S \rho) + \bra{\phi}\rho \ket{\phi} = \sum_{i} p_i
    \left(\Tr(\Pi_S \kera{\psi_i}) + \bra{\phi} (\kera{\psi_i}\right) \ket{\phi})
    \leq 1 + \sin \theta \, .    
  \end{equation*}

\end{proof}

\begin{proof}[Proof of Theorem~\ref{thm:one_main}]

The intuition behind the soundness proof is that the two branches (conditioning on register $\Pi$) of
  verifier's verification are competing and the prover cannot cheat one without also changing the other.  When the input $x\notin L$, 
  a cheating prover must apply an operation far from $O_{f^{-1}}$ on $\ket{Q}$
  to make $R$ accept.  We will show that when it applies such an operation, it must move
  the trap state $\ket{T}$ far from the correct state $\ket{T^H}$ which
  will be detected by the verifier.  
Now, we can finish the proof by showing the completeness and soundness.  

We introduce some notation first.  Let the state of the entire
  system after the prover's action be $\frac{1}{\sqrt{2}}(\ket{\psi_0}_{PMV}\ket{0}_B +
    \ket{\psi_1}_{PMV}\ket{1}_B) \, .$

  If the prover is honest, then $\ket{\psi_0} = \ket{0}_P\ket{Q^H}, \quad   \ket{\psi_1} =
    \ket{0}_P \ket{T^H}$.

   If the prover is dishonest, we can always assume that $O_{f^{-1}}$ is applied honestly, followed by an arbitrary unitary $\malp$ on its work register $P$ and message register $M$.  In this case
  \begin{equation*}
    \ket{\psi_0} = \malp\otimes I_V (\ket{0}_P\ket{Q^H}_{MV}), \quad \ket{\psi_1} =
    \malp\otimes I_V (\ket{0}_P\ket{T^H}_{MV}) \, .
  \end{equation*}
  For ease of notation, define $\rho_0 := \Tr_P(\kera{\psi_0}_{PMV})$ and $\rho_1 : =
    \Tr_P(\kera{\psi_1}_{PMV})$. 
    
    Let $\projr$ be the projection to the acceptance subspace
  $\accs \subseteq \hilbert_M \otimes \hilbert_V$ induced by
  $R$.  Observe that the verifier accepts with probability
  \begin{eqnarray*}
    \succp &:=& \frac 1 2 (p_0 + p_1) \, , \quad\text{where }  p_0= \Tr(\projr \rho_0)\, , \quad p_1 = \bra{T^H} \rho_1 \ket{T^H} \, .
  \end{eqnarray*}

  \mypar{Completeness} If $x\in \bar L$, then $\rho_0 = \kera{Q^H}$
  and $\rho_1 = \kera{T^H}$.  Therefore,
  $p_0 = \Tr(\projr \rho_0) \geq 1 - \veps$ by our hypothesis on the
  reduction.  Meanwhile $p_1 = \bra{T^H} \rho_1 \ket{T^H} = 1$.
  Therefore $\succp = \frac 1 2 (p_0 + p_1) \geq 1 - \veps/2$.

  \mypar{Soundness} Suppose that $x \notin \bar L$.  By
  Lemma~\ref{claim:u_indistinguishable}, we have that $p_1 = \bra{T^H} \rho_1 \ket{T^H}  = \bra{Q^H} \rho_0 \ket{Q^H}$.  
  Therefore $ \succp = \frac{1}{2} (p_0 + p_1) = \frac 1 2 (\Tr(\projr \rho_0) +
    \bra{Q^H}\rho_0\ket{Q^H}) \, .$
  Since $x\notin \bar L$, we know that
  $\reduction\cnot \ket{Q^H} = \sqrt{\delta}\ket{0}\ket{\phi_{x,0}} +
  \sqrt{1- \delta}\ket{1}\ket{\phi_{x,1}}$ with $\delta \leq
  \veps$.  Therefore, $\bra{Q^H} \Pi_R \ket{Q^H} \leq \veps$, i.e.,
  $\ket{Q^H}$ is almost orthogonal to the acceptance subspace
  $\accs$.  Then from the prover's perspective, to maximize the verifier's accepting probability, it needs to find a state whose projection
  on $\ket{Q^H}$ and $\accs$ combined is maximized.  By
  Lemma~\ref{lemma:maxproj}, the maximum is achieved by 
  a state bisecting $\ket{Q^H}$ and its projection on $\accs$, 
  and we conclude
  that $\succp = \frac{1}{2}(\Tr(\projr \rho_0) +
    \bra{Q^H} \rho_0 \ket{Q^H}) \leq \frac 1 2 (1 + \sqrt{\veps}).$
  
\end{proof}

Let $L$ be a hard problem in \class{NP} or \class{QMA}, we derive the following corollaries from Theorem~\ref{thm:one_main}.

\begin{corollary}\label{cor:1}
  If there exists a uniform one-query locally quantum reduction
  from a worst-case \class{NP}-hard decision problem to
  inverting a one-way permutation, then  \class{coNP} $\subseteq$
  \class{QIP$(2)$}.
\end{corollary}
\begin{proof}
  Suppose $L$ is \class{NP}-hard, and it reduces to Inv-OWP via a uniform one-query quantum locally random
  reduction.  By Theorem~\ref{thm:one_main},
  $\overline{L}\in \class{QIP(2)}$, hence
  $coNP \subseteq \class{QIP(2)}$.
\end{proof}

\begin{corollary}\label{cor:2}
If there exists a uniform one-query locally quantum
  reduction from a worst-case promise problem which is \class{QMA}-hard to
inverting a one-way permutation, then
\class{coQMA} $\subseteq$ \class{QIP$(2)$} 
\end{corollary}
\begin{proof}
  Suppose $L$ is \class{QMA}-hard and there exists a uniform one-query quantum locally random
  reduction from $L$ to Inv-OWP.  By Theorem~\ref{thm:one_main}, $\overline{L}\in \class{QIP(2)}$.  This implies $coQMA \subseteq \class{QIP(2)}$.  
\end{proof}

\subsection{Uniform non-adaptive locally quantum
  reductions}~\label{sec:ff_reduction}

In this section, we are going to generalize Theorem~\ref{thm:one_main}
such that the existence of a \emph{multi-query} uniform non-adaptive
locally quantum reduction with constant error implies
coNP$\subseteq$QIP($2$).

Let $f$ be a one-way permutation, and let $U_f$ be a circuit computing
it.  A uniform non-adaptive locally quantum reduction $(G,R)$ from a decision problem to the task of inverting $f$ is 
 defined as:  
\begin{eqnarray}
  \ket{x, 0} &\stackrel{\generator}{\longrightarrow}&
  \frac{1}{\sqrt{2^{mk}}} \sum_{q_1,\dots,q_k\in \mathbb{Z}_2^m} \ket{q_1,0,w_{x,1}(q_1)}\otimes\cdots\otimes\ket{q_k,0,w_{x,k}(q_k)}\nonumber\\
  &\stackrel{O_{f^{-1}}}{\longrightarrow}& \frac{1}{\sqrt{2^{mk}}} \sum_{q_1,\dots,q_k\in \mathbb{Z}_2^m} \ket{q_1,f^{-1}(q_1),w_{x,1}(q_1)}\otimes\cdots\otimes\ket{q_k,f^{-1}(q_k),w_{x,k}(q_k)} \nonumber\\
  &\stackrel{\reduction}{\longrightarrow}&
  a_0|0\rangle|\psi_{x,0}\rangle+a_1|1\rangle|\psi_{x,1}\rangle,\nonumber
\end{eqnarray}
where $|a_1|^2\geq2/3$ if $x\in L$ and $|a_1|^2\leq 1/3$ if $x\notin L$ is the probability the reduction accepts.  

\begin{theorem}~\label{thm:m_main}
Suppose there exists a uniform non-adaptive locally quantum reduction $(G,R)$ from a worst-case decision problem $L$ 
to $\iowp$.  Then, there exists a $\class{QIP}(2)$ protocol 
with completeness $1-\epsilon/2$ and soundness $1/2+2\sqrt{\epsilon}$ for 
$\overline{L}$, where $\epsilon$ is negligible.
\end{theorem}

Before giving the main theorem, we first show that the error of locally quantum reductions and quantum worst-case to
  average-case reductions can be reduced by parallel repetition.  

\begin{lemma}[Error reduction]~\label{lem:UQWC_amplification} The
  error of locally quantum reductions and quantum worst-case to
  average-case reductions can be reduced to an exponential small
  parameter $\epsilon$ in polynomial time and polynomial number of
  queries.
\end{lemma} 

\begin{proof}
  The error of both reductions can be reduced by parallel repetition.  The new reduction ($R',G'$) is described as follows: 
\begin{enumerate}
\item $G'$ operates $G$ $t$ times to generate $t$ copies of
  $|Q_{x,1}\rangle\otimes\cdots\otimes|Q_{x,k}\rangle$ and send all
  copies to the oracle in parallel, where $t$ is polynomial in the input length $n$.  
\item After getting all $t$ responses $|Q_{x,1}^H,\dots,Q_{x,k}^H\rangle$ from the oracle, $R'$ operates $R$ $t$ times and make the majority vote.  If more than $t/2$ copies are accepted, $R'$ accepts; otherwise, $R'$ rejects.   
\end{enumerate}

For completeness, the probability that $(G',R')$ rejects is
$\sum_{u<\frac{t}{2}} {{t}\choose{u}}
\frac{2}{3}^{u}(1-\frac{2}{3})^{t-u}$.  For soundness, the probability
that $(G',R')$ accepts is
$\sum_{u>\frac{t}{2}} {{t}\choose{u}}
\frac{1}{3}^{u}(1-\frac{1}{3})^{t-u}$.  Both are negligible. This completes the proof.
\end{proof}

Then we will show that such reduction does not exist unless
 \class{coNP} $\subseteq$ \class{QIP($2$)}.

\begin{proof}[Proof of Theorem~\ref{thm:m_main}]
  The error can be reduced to an exponentially small parameter $\epsilon$ by applying Lemma~\ref{lem:UQWC_amplification}.  The idea of the protocol for multiple queries is the same as the protocol in Protocol~\ref{fig:protocol1query} for one query.  The verifier generates a superposition of the query state and the trap state and sends part of the state to the prover.  In the following, we will give a \class{QIP}($2$) protocol which is similar to the protocol in
  Protocol~\ref{fig:protocol1query} for $\overline{L}$.    
  
  By Lemma~\ref{lem:UQWC_amplification}, the error of a quantum locally random $(G,R)$ can be reduced to an exponentially small parameter $\epsilon$ by parallel repetition, where we suppose $G$ is operated $t$ times and each time it generates $k$ queries.  We denote the new reduction as $(G',R')$.
  
  We now introduce the query state and the trap state the verifier generates. By applying $G'$ and $C$, the verifier generates $|Q_{1,1}\rangle|Q_{1,2}\rangle\otimes\cdots\otimes |Q_{t,k}\rangle$, where
  \[
    |Q_{i,j}\rangle =  \frac{1}{\sqrt{2^{m}}}\sum_{q\in \mathbb{Z}_2^{m}} |q,0\rangle|w_{x,j}(q),q\rangle \mbox{ for }1\leq i\leq t, 1\leq j\leq k.  
  \]
  Note that $i$ indicates the $i$-th copy generated from the parallel repetition in Lemma~\ref{lem:UQWC_amplification}.  Also, the verifier generates $|T\rangle^{\otimes tk}$, 
    where $|T\rangle$ is defined in Equation~\ref{eq:T}.  
    
    Then, we rearrange the qubits such that the first two registers of all $|Q_{i,j}\rangle$ and $|T\rangle$ are moved to the beginning in sequence as follows: 
    \begin{align}
    &|Q_{1,1},\dots, Q_{t,k}\rangle \rightarrow |\qx\rangle_{MV}
    = \frac{1}{2^{mkt/2}}\sum_{\hat{q}\in \mathbb{Z}_2^{mkt}} |\hat{q}, 0\rangle_M|w_{x}(\hat{q}),\hat{q}\rangle_V \label{eq:rearrange_Q}\\
    &|T\rangle^{\otimes k} \rightarrow |\hat{T}\rangle_{MV}
    = \frac{1}{2^{mkt/2}}\sum_{\hat{q}\in \mathbb{Z}_2^{m k t}} |\hat{q}, 0\rangle_M|0,\hat{q}\rangle_V.  \label{eq:reqrrange_T}
    \end{align}
    where $\hat{q}=[q_{1,1},\dots, q_{1,k},\dots,q_{t,1},\dots,q_{t,k}]$ and $w_x(\hat{q}) = [w_{x,1}(q_{1,1}),\dots,w_{x,k}(q_{t,k})].$
    For example, given a state of two queries $\sum_{q,q'} |q,0\rangle|w_{x,1}(q),q\rangle|q',0\rangle|w_{x,2}(q),q'\rangle$, following the rearrangement, we represent it as $\sum_{q,q'} |qq',0\rangle|w_{x,1}(q)w_{x,2}(q'),qq'\rangle$.  

    Similarly, we define 
    \begin{align*}
        &|\qx^H\rangle_{MV} = \frac{1}{2^{mkt/2}}\sum_{\hat{q}\in \mathbb{Z}_2^{mkt}} |\hat{q},f^{-1}(\hat{q})\rangle_M|w_{x}(\hat{q}),\hat{q}\rangle_V \\
      &|\hat{T}^H\rangle_{MV} = \frac{1}{2^{mkt/2}}\sum_{\hat{q}\in \mathbb{Z}_2^{mkt}} |\hat{q},f^{-1}(\hat{q})\rangle_M| 0,\hat{q}\rangle_V, \nonumber
    \end{align*}
    where $f^{-1}(\hat{q})= (f^{-1}(q_{1,1}),\dots,f^{-1}(q_{t,k})).$ 
    
\begin{algorithm}[h]
  \begin{mdframed}[style=figstyle,innerleftmargin=10pt,innerrightmargin=10pt]
  The protocol takes place in the space
      $\mathcal{H}_P \otimes \mathcal{H}_{M} \otimes \mathcal{H}_V
      \otimes \mathcal{H}_{\Pi}$ 
      where $P$ is the private register of the prover, $M$ is the
      register exchanged between the prover and the verifier, and
      $V$ and $\Pi$ are registers which are private to the verifier.
    \begin{enumerate}
    \item \emph{The verifier's query.}  The verifier
      prepares the state below.  The message register $M$ is sent to
      the prover.
    \begin{eqnarray*}
       \ket{\hat{S}}_{MV\Pi} :=
        \frac{1}{\sqrt{2}}(|\hat{Q}\rangle_{MV}|0\rangle_{\Pi}+|\hat{T}\rangle_{MV}|1\rangle_{\Pi}).\nonumber
    \end{eqnarray*}
    
  \item \emph{The prover's response.}  The prover applies some unitary $U_{PM}$ on register $M$ and its private register $P$ and sends the message register back to the verifier.  
  \item \emph{The verifier's verification.} The verifier applies
    $C$ to erase $\hat{q}$ in $V$.  The verifier then measures $\Pi$ to obtain
    $b \in \{0,1\}$, and does the following:
    \begin{itemize}
    \item \text{(Computation verification)} If $b=0$, apply $R'$ on $MV$
      and measure the output qubit.  Accept if the outcome is $0$.
    \item \text{(Trap verification)} If $b=1$, apply $V_T$ on each $|T\rangle$ in $MV$
      and measure $MV$.  Accept if the outcome is the all $0$ string.
    \end{itemize}
  \end{enumerate}
  \caption{\class{QIP}($2$) protocol for $\overline{L}$ using a non-adaptive locally quantum reduction.} 
  \label{fig:protocolmquery}
  \end{mdframed}
\end{algorithm}

    The \class{QIP($2$)} protocol for $\overline{L}$ is shown in Protocol~\ref{fig:protocolmquery}.  Note that the prover's behavior $U_{PM}$ can be represented as $U'_{PM}O_{f^{-1}}$ where $U'_{PM}$ is an arbitrary unitary a cheating prover may apply.  In the following, we show that the protocol in Protocol~\ref{fig:protocolmquery} is a \class{QIP($2$)} protocol for $\overline{L}$ .
 
For the completeness condition, when $x\in \overline{L}$, the verifier accepts with probability $\geq 1-\epsilon/2$ via the same calculation in Section~\ref{sec:one_reduction}.  

For the soundness condition, assume $x\notin \overline{L}$.  Let 
\[
\hat{\sigma}^{U'}_Q=\Tr_P(U'_{PM}\otimes I_V(|0\rangle\langle 0|\otimes |\hat{Q}^H\rangle\langle \hat{Q}^H|)U'^{\dagger}_{PM}\otimes I_V)
\]
\[
    \hat{\sigma}^{U'}_T=\Tr_P(U'_{PM}\otimes I_V(| 0\rangle\langle  0|\otimes |\hat{T}^H\rangle\langle \hat{T}^H|)U'^{\dagger}_{PM}\otimes I_V).  
\]
Since $|\hat{T}^H\rangle$ and $|\hat{Q}^H\rangle$ are two purifications of the mixed state $\Tr_{V}(\ket{\hat{Q}^H}\bra{\hat{Q}^H})$ on register $M$, 
\begin{eqnarray}
\bra{\hat{Q}^H}\hat{\sigma}^{U'}_{Q}\ket{\hat{Q}^H}= \bra{\hat{T}^H}\hat{\sigma}^{U'}_{T}\ket{\hat{T}^H} \label{eq:m_equal}
\end{eqnarray}
by Lemma~\ref{lemma:compm}.  Then, we do a similar calculation as in the proof of soundness for Theorem~\ref{thm:one_main}, which gives an upper bound $1/2+\sqrt{\epsilon}$ on the probability that the verifier accepts.
\end{proof}

The following two corollaries follow from Theorem~\ref{thm:m_main}, which proofs are the same as the proof for Corollary~\ref{cor:1} and Corollary~\ref{cor:2}.

\begin{corollary}
If there exists a uniform non-adaptive quantum locally random
  reduction from a worst-case decision problem which is \class{NP}-hard to
the task of inverting a one-way permutation, then  \class{coNP}
$\subseteq$ \class{QIP$(2)$}.
\end{corollary}

\begin{corollary}
If there exists a uniform non-adaptive quantum locally random
  reduction from a worst-case promise problem which is \class{QMA}-hard to
the task of inverting a one-way permutation, then
\class{coQMA} $\subseteq$ \class{QIP$(2)$}.
\end{corollary}

%----------------------------------------%
\section{Generalizations}
\label{sec:gen}
%----------------------------------------%

In this section, we will show that we can generalize our techniques to other settings.  We will first see that we can deal with distributions which are not uniform but close to uniform. Then, we will show that in addition to locally quantum reduction, we can also handle the case of quantum worst-case to average-case reduction when the queries are ``smooth''. Finally, we will consider the task of inverting a regular one-way function and will show that when the oracle can do a quantum sampling of the preimages, then the same theorem for Inv-OWP also holds for Inv-OWF.

%----------------------------------------%
\subsection{Smooth locally quantum reductions to $\iowp$}
\label{sec:smooth}
%----------------------------------------%
We first show that when the distribution is not far from uniform distribution, we can still construct a \class{QIP}($2$) protocol given the non-adaptive quantum locally reduction. Therefore, the existence of such reduction also imply $ \class{coNP}\subseteq \class{QIP(2)}$.

The difficulty to apply Protocol~\ref{fig:protocol1query} and Protocol~\ref{fig:protocolmquery} to non-uniform distributions is that we do not know how to construct a trap state that can be mapped to $|0\rangle$ efficiently and has the state in the message register be indistinguishable from the actual query state.  Here we show that if the distribution is smooth-computable, then the verifier can use the same trap state by applying quantum rejection sampling~\cite{OMR13} to prevent the prover from cheating.  

\begin{definition}[Smooth-computable distributions]\label{def:smooth_dist}
A distribution $\mathcal{D} = \{\mathcal{D}_n: n\in \mathbb{N}\}$ is said to be smooth-computable if it satisfies the following properties.  Let $d_q = \Pr[q\sim \mathcal{D}_{n}]$, $d_{min, n}= \min_{q\in \{0,1\}^n} d_q$ and $d_{max, n}= max_{q\in \{0,1\}^n} d_q$.  
\begin{enumerate}
    \item For $n\in \mathbb{N}$, for all $q$ where $|q| = n$, the function $f_n: f_n(q) = d_q$ is polynomial-time computable.  
    \item For $n\in \mathbb{N}$, $2^n d_{min, n} \geq \frac{1}{poly(n)}$ and $2^n d_{max, n} \leq poly(n)$.
\end{enumerate}
\end{definition}

Loosely speaking, smooth-computable distributions are point-wise close
to the uniform distribution.  It is worth noting that
Protocol~\ref{fig:protocolmquery} can handle those that have
negligible statistical distance to the uniform distribution.  However,
there exists some smooth-computable distribution that has
inverse-polynomial distance from the uniform distribution.  In such
cases, it is unclear if soundness still holds in
Protocol~\ref{fig:protocolmquery}.

Again, we start with the special case of \emph{one-query}
reductions with negligible error.  Generalizing to multiple non-adaptive queries is similar
to the case of uniform distributions.

Let $f$ be a one-way permutation on $\bool^n$, and let $U_f$ be a quantum circuit computing it.  
A smooth one-query locally
quantum reduction according to a smooth-computable distribution
$\mathcal{D}=\{\mathcal{D}_n: n\in \mathbb{N}\}$ proceeds as follows:

\begin{eqnarray}
  \ket{x, 0}\ket{0} &\stackrel{\generator}{\longrightarrow}&
                                                      \frac{1}{\sqrt{2^{m}}} \sum_{q\in\mathbb{Z}_2^m} \sqrt{d_q}\ket{q,0}\ket{w_{x}(q)}\label{eq:slqr_G}\\
             &\overset{O_{f^{-1}}}{\longrightarrow}& \frac{1}{\sqrt{2^{m}}} \sum_{q\in \mathbb{Z}_2^m} \sqrt{d_q}\ket{q,f^{-1}(q)}\ket{w_{x}(q)}\\ 
             &\stackrel{\reduction}{\longrightarrow}&a_0|0\rangle|\psi_{x,0}\rangle+a_1|1\rangle|\psi_{x,1}\rangle,\label{eq:slqr_R}
\end{eqnarray}
where $|a_1|^2\geq 1-\epsilon$ if $x\in L$ and $|a_1|^2\leq \epsilon$ if $x\notin L$ is the probability the reduction accepts and $d_q$ is the probability $q$ is drawn from $\mathcal{D}_{n}$ for $n=|q|$.

\begin{theorem}~\label{thm:one_smooth_main} Suppose there exists a one-query
  smooth locally quantum reduction with exponentially small error
  $\epsilon$ from a worst-case decision problem $L$ to the task of
  inverting a polynomial-time computable permutation.  Then there
  exists a $\class{QIP}(2)$ protocol with completeness $1-\epsilon/2$ and
  soundness $1/2+2\sqrt{\epsilon}$ for $\overline{L}$
\end{theorem}

The proof of Theorem~\ref{thm:one_smooth_main} relies on the quantum rejection sampling technique
in~\cite{OMR13}. We give the definition of the quantum rejection sampling problem and we adapt their tools in the Lemma~\ref{lem:QRSamp}. 
\begin{definition} [Quantum rejection sampling problem $QRSP(\mathcal{D}, \mathcal{D}', n)$]\label{def:QRSampling}
Given an oracle $O_{\mathcal{D}}: \ket{0} \rightarrow \sum_{x=1}^{2^n} \sqrt{d_x}\ket{\xi_x}\ket{x}$ as a unitary, where $d_x\sim \mathcal{D}_{n}$ and $\ket{\xi_x}$ are some unknown fixed states.  The Quantum rejection sampling problem is to prepare the state $\sum_{x=1}^{2^n} \sqrt{d'_x}\ket{\xi_x}\ket{x}$ for $d'_x\sim \mathcal{D}'_{n}$.  
\end{definition}

\begin{lemma}~\label{lem:QRSamp} Let
  $\mathcal{D}= \{\mathcal{D}_n: n\in \mathbb{N}\}$ be a
  smooth-computable distribution and $\mathcal{U}$ be the uniform
  distribution.  There exists a quantum polynomial-time algorithm
  $QRSampling(\mathcal{D}\rightarrow \mathcal{U})$ that takes
  $\gamma = (\lceil\frac{1}{2^n d_{min,n}}\rceil)^2$ copies of
  $\sum_{x=1}^{2^n} \sqrt{d_x}\ket{\xi_x}\ket{x}$ and outputs a state
  that has negligible trace distance $\delta$ to
  $\sum_{x=1}^{2^n} \sqrt{\frac{1}{2^n}}\ket{\xi_x}\ket{x}$.  Similarly
  $QRSampling(\mathcal{U}\rightarrow \mathcal{D})$ takes
  $\gamma' = (\lceil2^n d_{max,n}\rceil)^2$ copies of
  $\sum_{x=1}^{2^n} \sqrt{\frac{1}{2^n}}\ket{\xi_x}\ket{x}$ and
  outputs a state that has negligible trace distance $\delta'$ to
  $\sum_{x=1}^{2^n} \sqrt{d_x}\ket{\xi_x}\ket{x}$.
\end{lemma}

Note that $\gamma$ and $\gamma'$ are polynomial in $n$ when $\mathcal{D}$ is smooth according to Definition~\ref{def:smooth_dist}.

\begin{proof}

We first show the sample complexity.  It has been shown in~\cite{OMR13} that Algorithm~\ref{alg:QRSampling} can solve the $QRSP(\mathcal{D}, \mathcal{D}', k)$ exactly with $1-e^{-\beta}$ with $\beta^2$ samples generated from $O_{\mathcal{D}}$ for $\frac{1}{\beta} = min_{x} d_x/d_x'$.  In case $\mathcal{D}=\mathcal{U}$, we have $\frac{1}{\beta} = min_x \frac{1}{2^nd_x}=\frac{1}{2^nd_{max,n}}$.  In case $\mathcal{D}'=\mathcal{U}$, we have $\frac{1}{\beta} = 2^nd_{min,n}$.

Algorithm~\ref{alg:QRSampling} can also be done in polynomial time.  Consider the case where $\mathcal{D}'=\mathcal{U}$.  The Step 2 in Algorithm~\ref{alg:QRSampling} can be viewed as a control rotation on the first and the third register.  
\[
S=\sum_{i=1}^{2^n}\frac{1}{d_i}
\begin{bmatrix}
    \sqrt{d_i-\frac{1}{ 2^n\gamma}}      & -\sqrt{\frac{1}{ 2^n\gamma}} \\
    \sqrt{\frac{1}{ 2^n\gamma}}      & \sqrt{d_i-\frac{1}{2^n\gamma}}
\end{bmatrix}
\otimes I\otimes \ket{i}\bra{i}
\]
By Solovay-Kitaev theorem, any known one-qubit unitary $V$ can be approximated by $V'$ which is implemented by polynomial number of gates from a finite universal gate set with an exponentially small error $\delta = max_{\ket{\psi}} \|(V - V')\ket{\psi}\|$.  Since $\mathcal{D}$ and $\mathcal{U}$ are polynomial-time computable as in Definition~\ref{def:smooth_dist}, we can approximate $S$ in polynomial time.  This completes the proof.  The analysis for the case where $\mathcal{D}=\mathcal{U}$ is the same.  

\floatname{algorithm}{Algorithm}
\begin{algorithm}
\caption{$QRSampling(\mathcal{D}\rightarrow \mathcal{D}')$}\label{alg:QRSampling}
\begin{algorithmic}[1]
\State Let $\frac{1}{\beta} = \min_x \frac{d_x}{d'_x}$.
\State Apply $O_{\mathcal{D}}$ to generate $\sum_{x=1}^{2^n} \sqrt{d_i}\ket{\xi_i}\ket{i}$.  
\State Pick $\vec{\alpha}\in \mathbb{R}^{2^n}_+$ where $\alpha_i = \frac{d'_i}{\beta}$and rotate the state in the first register by $S$
        \[
            S: \ket{0}(\sum_{x=1}^{2^n} \sqrt{d_i}\ket{\xi_i}\ket{i}) \rightarrow \sum_{x=1}^{2^n} (\sqrt{d_i-\alpha_i}\ket{0}+\sqrt{\alpha_i}\ket{1}) \ket{\xi_i}\ket{i}.
        \]
\State Measure the first qubit, which gives $\sum_{i=1}^{2^n} \sqrt{d'_i}\ket{\xi_i}\ket{i}$ with probability $\frac{1}{\beta}$.
\State By repeating steps 2 to 4 $\Theta(\beta^2)$ times, one can prepare the state $\sum_{i=1}^{2^n} \sqrt{d'_i}\ket{\xi_i}\ket{i}$ with probability $1-e^{-\beta}$ .  
\end{algorithmic}
\end{algorithm}

\end{proof}

In the following, we first give a new protocol Protocol~\ref{fig:protocolsquery} in Section~\ref{sec:protocol_qrs} and then prove that the protocol is a \class{QIP}($2$) protocol.

\subsubsection{The new protocol for $\overline{L}$ using quantum
  rejection sampling.}~\label{sec:protocol_qrs} We first describe some
states which are used in the protocol.
    \begin{eqnarray}
    |Q_{\mathcal{D}}\rangle_{MV} = \frac{1}{\sqrt{2^n}}\sum_{q\in \mathbb{Z}_2^m} \sqrt{d_q}|q,0\rangle_M|w_x(q),q\rangle_V, \label{eq:SQ}
    \end{eqnarray}
    where $|Q_{\mathcal{D}}\rangle_{MV}$ without the copy of $q$ in $V$ is the query state generated from $G$ as in Equation~\ref{eq:slqr_G}.  
    \begin{eqnarray}
    |Q^H_{\mathcal{D}}\rangle_{MV} = \frac{1}{\sqrt{2^n}}\sum_{q\in \mathbb{Z}_2^m} \sqrt{d_q}|q,f^{-1}(q)\rangle_M|w_x(q),q\rangle_V, \label{eq:SQ_2}
    \end{eqnarray}    
    where $|Q^H_{\mathcal{D}}\rangle$ without the extra copy $q$ in register $V$
    is the state the actual reduction $R$ gets after querying the
    oracle as in Equation~\ref{eq:slqr_R}.  
    
    By applying $QRSampling(\mathcal{D}\rightarrow \mathcal{U})$,  one can prepare the state $\ket{\tilde{Q}}$ from $|Q_{\mathcal{D}}\rangle_{MV}$ such that $D(\ket{\tilde{Q}}, \ket{Q})\leq \delta$, where $\delta$ is an exponentially small error.  We also define $\ket{\tilde{Q}^H} := O_{f^{-1}}\ket{\tilde{Q}}$.  Similarly, one can prepare the state $\ket{\tilde{Q}^H_{\mathcal{D}}}$ from $\ket{\tilde{Q}^H}$ by $QRSampling(\mathcal{U}\rightarrow \mathcal{D})$ such that $D(\ket{\tilde{Q}^H_{\mathcal{D}}}, \ket{Q^H_{\mathcal{D}}})\leq \delta'$, where $\delta'$ is exponentially small.

We give Protocol~\ref{fig:protocolsquery} for $\overline{L}$ with
non-adaptive smooth locally quantum reductions and proves Theorem~\ref{thm:one_smooth_main}. 
\ifcccsub
\begin{proof}[Proof of Theorem~\ref{thm:one_smooth_main}]

  Let the state of the entire system after the prover's action be
  \begin{equation*}
    \frac{1}{\sqrt{2}}(\ket{\tilde{\psi}}_{PM_1\dots M_{\gamma'}V_1\dots V_{\gamma'}}\ket{0}_B +
    \ket{\phi}_{PM_1\dots M_{\gamma'}V_1\dots V_{\gamma'}}\ket{1}_B) \, .
  \end{equation*}
To simplify the notation, we let $M=M_1M_2\dots M_{\gamma'}$ and $V=V_1V_2\dots V_{\gamma'}$.  
  If the prover is honest, then
  \begin{equation*}
    \ket{\tilde{\psi}} = \ket{0}_P\ket{\tilde{Q}^H}\otimes\cdots\otimes \ket{\tilde{Q}^H}, \quad   \ket{\phi} =
    \ket{0}_P \ket{T^H}\otimes\cdots\otimes \ket{T^H}\, ,
  \end{equation*}
where $F(\ket{\tilde{Q}^H},\ket{Q^H})\geq 1-\delta$ according to Lemma~\ref{lem:QRSamp}.  
  If the prover is dishonest, we can always assume that the prover first
  applies $O_{f^{-1}}$ honestly and then applies an arbitrary unitary
  $\malp$ on its work register $P$ and message register $M$.  In this
  case
  \begin{equation*}
    \ket{\tilde{\psi}} = \malp_{PM}\otimes I_V (\ket{0}_P\ket{\tilde{Q}^H}\otimes\cdots\otimes \ket{\tilde{Q}^H}), \quad \ket{\phi} =
    \malp_{PM}\otimes I_V (\ket{T^H}\otimes\cdots\otimes \ket{T^H}) \, .
  \end{equation*}
  For ease of notation, we define
  \begin{equation*}
    \tilde{\rho}_0 := \Tr_P(\kera{\tilde{\psi}}); \quad \rho_1 : =
    \Tr_P(\kera{\phi}) \, .
  \end{equation*}

  Let $\projr$ be the projection to the acceptance subspace
  $\accs \subseteq \hilbert_{M} \otimes \hilbert_{V}$ induced by
  the verifier's verification.  Observe that the verifier accepts with probability
  \begin{eqnarray*}
    \succp := \frac 1 2 (p_0 + p_1) \, , \quad\text{where }  p_0= \Tr(\projr \tilde{\rho}_0)\, , \quad p_1 = \bra{T^H}^{\otimes \gamma'} \rho_1 \ket{T^H}^{\otimes \gamma'} \, .
  \end{eqnarray*}

  \mypar{Completeness} If $x\in \bar L$, then $\tilde{\rho}_0 = \kera{\tilde{Q}^H}$
  and $\rho_1 = \kera{T^H}$.  Therefore,
  $p_0 = \Tr(\projr \tilde{\rho}_0) \geq 1 - \veps - 2\delta$ where $\epsilon$ is from our hypothesis on the
  reduction and $\delta$ is the error from the quantum rejection sampling.  Meanwhile $p_1 = \bra{T^H} \rho_1 \ket{T^H} = 1$.  Therefore $\succp = \frac 1 2 (p_0 + p_1) \geq 1 - (\veps+2\delta)/2$.

  \mypar{Soundness} Suppose that $x \notin \bar L$.  Let $\ket{\psi} = \malp_{PM}\otimes I_V (\ket{0}_P\ket{Q^H}\otimes\cdots\otimes \ket{Q^H})$ and $\rho_0 := \Tr_P(\kera{\psi})$.
  By
  Lemma~\ref{claim:u_indistinguishable}, we have that
  \begin{equation*}
    \bra{T^H}^{\otimes \gamma'} \rho_1 \ket{T^H}^{\otimes \gamma'}  = \bra{Q^H}^{\otimes \gamma'} \rho_0 \ket{Q^H}^{\otimes \gamma'} \,  , 
  \end{equation*}
  and then we are going to show that $\bra{T^H}^{\otimes \gamma'} \rho_1 \ket{T^H}^{\otimes \gamma'}$ is close to $\bra{Q^H}^{\otimes \gamma'} \tilde{\rho}_0 \ket{Q^H}^{\otimes \gamma'}$ except for a exponentially small error.
  
  First, by monotonicity of the fidelity, $F(\rho_0, \tilde{\rho}_0) \geq F(\ket{Q^H},\ket{\tilde{Q}^H})\geq 1-\delta$.  Then we define the angles between states $\ket{Q^H}^{\otimes \gamma'}$, $\rho_0$ and $\tilde{\rho}_0$ as \[
    A(\ket{Q^H}^{\otimes \gamma'}, \rho_0) = \arccos{F(\ket{Q^H}^{\otimes \gamma'}, \rho_0)}, \quad A(\tilde{\rho}_0, \rho_0) = \arccos{F(\tilde{\rho}_0, \rho_0)}, \text{ and}
  \]
  \[
    A(\ket{Q^H}^{\otimes \gamma'}, \tilde{\rho}_0) = \arccos{F(\ket{Q^H}^{\otimes \gamma'}, \tilde{\rho}_0)}.  
  \]
  By the triangular inequality, 
  \[
    A(\ket{Q^H}^{\otimes \gamma'}, \tilde{\rho}_0) \leq A(\ket{Q^H}^{\otimes \gamma'}, \rho_0) + A(\tilde{\rho}_0, \rho_0).  
  \]
  This gives 
  \begin{eqnarray*}
    F(\ket{Q^H}^{\otimes \gamma'}, \tilde{\rho}_0) &\geq& \cos{(A(\ket{Q^H}^{\otimes \gamma'}, \rho_0) + A(\tilde{\rho}_0, \rho_0))}\\
    &=& F(\ket{Q^H}^{\otimes \gamma'}, \rho_0)F(\tilde{\rho}_0, \rho_0) - \sqrt{1-F(\ket{Q^H}^{\otimes \gamma'}, \rho_0)}\sqrt{1-F(\tilde{\rho}_0, \rho_0)}\\
    &\geq& F(\ket{Q^H}^{\otimes \gamma'}, \rho_0) - 2\sqrt{\delta}.
  \end{eqnarray*}
  We can also get an upper bound on $F(\ket{Q^H}^{\otimes \gamma'}, \tilde{\rho}_0)$ as follows: 
  By triangular inequality, 
  \[
    A(\ket{Q^H}^{\otimes \gamma'}, \tilde{\rho}_0) \geq A(\ket{Q^H}^{\otimes \gamma'}, \rho_0) - A(\tilde{\rho}_0, \rho_0),  
\]
which implies 
\begin{eqnarray*}
    F(\ket{Q^H}^{\otimes \gamma'}, \tilde{\rho}_0) &\leq& \cos{(A(\ket{Q^H}^{\otimes \gamma'}, \rho_0) - A(\tilde{\rho}_0, \rho_0))}\\
    &=& F(\ket{Q^H}^{\otimes \gamma'}, \rho_0)F(\tilde{\rho}_0, \rho_0) + \sqrt{1-F(\ket{Q^H}^{\otimes \gamma'}, \rho_0)}\sqrt{1-F(\tilde{\rho}_0, \rho_0)}\\
    &\leq& F(\ket{Q^H}^{\otimes \gamma'}, \rho_0) +\sqrt{\delta}.  
  \end{eqnarray*}
  
We can conclude that  
  \[
    \bra{Q^H}^{\otimes \gamma'} \tilde{\rho}_0 \ket{Q^H}^{\otimes \gamma'} = \bra{Q^H}^{\otimes \gamma'} \rho_1 \ket{Q^H}^{\otimes \gamma'}+c\sqrt{\delta}=\bra{T^H}^{\otimes \gamma'} \rho_1 \ket{T^H}^{\otimes \gamma'}+c\sqrt{\delta}
  \]
  for $c$ a small constant.  
  Therefore
  \begin{equation*}
    \succp = \frac{1}{2} (p_0 + p_1) = \frac 1 2 (\Tr(\projr \tilde{\rho}_0) +
    \bra{Q^H}^{\otimes \gamma'}\tilde{\rho}_0\ket{Q^H}^{\otimes \gamma'}+c\sqrt{\delta}) \, .
  \end{equation*}
 By Lemma~\ref{lemma:maxproj}, we can give an upper bound on $\succp$ as follows.  
  \begin{equation*}
    \succp = \frac{1}{2}(\Tr(\projr \rho_0) +
    \bra{Q^H} \rho_0 \ket{Q^H}) \leq \frac 1 2 (1 + \sqrt{\veps}+c\sqrt{\delta}) \, .
  \end{equation*}
  
\end{proof}
\fi
\floatname{algorithm}{Protocol}
\begin{algorithm}[ht]
\begin{mdframed}[style=figstyle,innerleftmargin=10pt,innerrightmargin=10pt]
  Let $\gamma = (\lceil\frac{1}{2^n d_{min}}\rceil)^2$ and $\gamma' = (\lceil2^n d_{max}\rceil)^2$, where $d_{min}= \min_{q\in \{0,1\}^n} \Pr[q\sim \mathcal{D}_{n}]$ and $d_{max}= max_{q\in \{0,1\}^n} \Pr[q\sim \mathcal{D}_{n}]$.
    \begin{enumerate}
    \item \emph{The verifier's query.}  The verifier
      prepares the state 
      \begin{eqnarray*}
       \ket{S}_{MV\Pi} :=
        \frac{1}{\sqrt{2}}(\ket{\tilde{Q}}_{M_1V_1}\otimes\cdots\otimes\ket{\tilde{Q}}_{M_{\gamma'}V_{\gamma'}}|0\rangle_{\Pi}+\ket{T}_{M_1V_1}\otimes\cdots\otimes\ket{T}_{M_{\gamma'}V_{\gamma'}}|1\rangle_{\Pi}).\nonumber
    \end{eqnarray*}
      The message registers $M_1,\dots,M_{\gamma'}$ are sent to
      the prover, and the verifier keeps $V_1,\dots,V_{\gamma'}$ and $\Pi$.  $|\tilde{Q}\rangle$ can be prepared from $\gamma$ copies of $|Q_{\mathcal{D}}\rangle$ by applying $QRSampling(\mathcal{D}\rightarrow \mathcal{U})$.  
  \item \emph{The prover's response.}  The prover applies some unitary $U_{PM_1\dots M_{\gamma'}}$ on registers $M_1\dots M_{\gamma'}$ and its private register $P$ and sends the message registers back to the verifier.  
  \item \emph{The verifier's verification.} The verifier applies
    $C$ to erase $q$ in $V_1\dots V_{\gamma'}$.  The verifier then measures $\Pi$ to obtain
    $b \in \{0,1\}$, and does the following:
    \begin{itemize}
    \item \text{(Computation verification)} If $b=0$, apply $QRSampling(\mathcal{U}\rightarrow\mathcal{D})$ to get a state $\ket{\tilde{Q}^H_{\mathcal{D}}}$,
    apply $R$ on $\ket{\tilde{Q}^H_{\mathcal{D}}}$ and measure the output qubit.  Accept if the outcome is $0$.
    \item \text{(Trap verification)} If $b=1$, apply $V_T$ on $M_iV_i$ for $i\in[\gamma']$
      and measure.  Accept if the outcome is all $0$.
    \end{itemize}
  \end{enumerate}
    \caption{\class{QIP}($2$) protocol for $\overline{L}$ with non-adaptive smooth locally quantum reductions.}
  \label{fig:protocolsquery}
  \end{mdframed}
  \end{algorithm}

By the same proof as in Section~\ref{sec:ff_reduction}, we generalize Theorem~\ref{thm:one_smooth_main} to Theorem~\ref{thm:m_smooth_main}.
\begin{theorem}~\label{thm:m_smooth_main} Suppose there exists a one-query
  smooth locally quantum reduction with constant error from a worst-case decision problem $L$ to $\iowp$.  Then there
  exists a $\class{QIP}(2)$ protocol with completeness $1-\epsilon/2$ and
  soundness $1/2+2\sqrt{\epsilon}$ for $\overline{L}$, where $\epsilon$ is negligible.  
\end{theorem}

\subsection{Non-adaptive quantum worst-case to average-case
  reductions}\label{sec:w2a_reduction}

The same idea above actually also works for the non-adaptive quantum
worst-case to average-case reduction defined in
Definition~\ref{def:qwar}.  We show that if the queries are generated
arbitrarily according to known smooth-computable distributions, i.e.,
the distributions of each query can be different but are
smooth-computable and known, then the existence of such reductions
also implies \class{coNP}$\subseteq$ \class{QIP}($2$).  We call this reduction
\emph{known smooth non-adaptive quantum worst-case to average-case
  reduction}.

\begin{theorem}~\label{thm:m_main_BT}
Suppose there exists a known smooth non-adaptive quantum worst-case to average-case reduction with average hardness $\delta$
$(G,R)$ from a worst-case decision problem $L$ to $\iowp$.  Then, there exists a \class{QIP($2$)} protocol 
with completeness $1-\epsilon/2$ and soundness $1/2+2\sqrt{\epsilon}$ for 
$\overline{L}$
\end{theorem}

\begin{proof}
Suppose $(G,R)$ is the reduction and $G$ generates $k$ uniform queries.  
Given any function $g$ which is $\delta$-close to $f^{-1}$ as Definition~\ref{def:close}.  Then, the smooth non-adaptive worst-case to average-case reduction is as follows: 
\begin{eqnarray}
|x, 0\rangle &\xrightarrow{G}& (\sum_{q}\sqrt{d_{1,q}}\ket{q,0,w_x(q)})\otimes\cdots\otimes(\sum_{q}\sqrt{d_{k,q}}\ket{q,0,w_x(q)})  \nonumber\\
&\xrightarrow{O_{g}}& (\sum_{q}\sqrt{d_{1,q}}\ket{q,f^{-1}(q),w_x(q)})\otimes\cdots\otimes(\sum_{q}\sqrt{d_{k,q}}\ket{q,f^{-1}(q),w_x(q)})  \nonumber\\
&\xrightarrow{R}& \sqrt{p}|L(x)\rangle|\psi_{x,0}\rangle + \sqrt{1-p}|1-L(x)\rangle|\psi_{x,1}\rangle, \nonumber
\end{eqnarray}
where $p\geq 2/3$ and $d_{i,q}$ are
the probability that $q$ is drawn from a smooth-computable distribution $\mathcal{D}^{(i)}_{|q|}$.  Note that $\mathcal{D}^{(i)}_{|q|}$ can be difference from $\mathcal{D}^{(j)}_{|q|}$ for $i\neq j$.  The error of the reduction can be reduced to an exponentially small parameter $\epsilon$ by Lemma~\ref{lem:UQWC_amplification}.  

Given such a reduction from $L$ to $\iowp$, Protocol~\ref{fig:protocolsquery} decides $\overline{L}$.  It is worth noting that since the distribution of each query is known and smooth-computable, one can apply quantum resampling for uniform distribution as in Protocol~\ref{fig:protocolsquery}.  For completeness, the honest prover always simulates $O_{f^{-1}}$, which is the same honest prover considered in Theorem~\ref{thm:one_smooth_main}.  Hence, the verifier accepts with probability at least $1-\frac{\epsilon}{2}$.  For soundness, if the prover's operation is $\delta$-close to $O_{f^{-1}}$, then the verifier accepts with probability $\leq (1+\epsilon)/2$.  Else if it chooses an operation $U'_{PM}$ which is not close to any $\delta$-close oracle for $O_{f^{-1}}$, then the modified trap state must be far from the original trap state.  By the calculation in Section~\ref{sec:smooth}, we get the same upper bound on the soundness.
\end{proof}

The following two corollaries follow from Theorem~\ref{thm:m_main_BT}. 
 
\begin{corollary}
If there exists a known smooth non-adaptive quantum worst-case to average-case reduction from a worst-case decision problem which is \class{NP}-hard to
$\iowp$, then  \class{coNP} $\subseteq$ \class{QIP$(2)$}
\end{corollary}

\begin{corollary}
If there exists a known smooth non-adaptive quantum worst-case to average-case reduction from a worst-case promise problem which is \class{QMA}-hard to
$\iowp$, then \class{coQMA} $\subseteq$ \class{QIP$(2)$}
\end{corollary}

\subsection{Fixed-preimage-sized functions and
  quantum-sampling oracles}~\label{sec:fpsf}

One interesting question to address is whether the protocol we give in
Section~\ref{sec:one_reduction} can be used for more general
functions. We observe that it indeed extends to functions which have
fixed preimage size and surjective (onto), e.g., k-to-1 functions, if
the oracle in the reduction is capable of quantum sampling from all
$k$ preimages. To be more specific, this reduction follows
Definition~\ref{def:qlrr} except that the response of each query
changes to
\[
|Q^H\rangle = \frac{1}{\sqrt{2^n}}\sum_{q\in \mathbb{Z}_2^m}
        (\ket{q}\frac{1}{\sqrt{|f^{-1}(q)|}}\sum_{z\in f^{-1}(q)}\ket{z})|w_x(q),q\rangle, 
\]
where the oracle gives a uniform superposition of all solutions to
$q$. We call such reductions a \textsl{locally quantum reduction with
  a quantum-sampling oracle}. This kind of oracle has been considered in
cryptography. For instance, Bacon et al.~\cite{BCD06} showed that the dihedral hidden subgroup problem reduces to quantum sampling subset sum solutions. Our result naturally generalizes when considering
this kind of reductions.

\begin{corollary}\label{cor:0}
If there exists a uniform one-query \textbf{locally quantum reduction with a quantum-sampling oracle}
  from a worst-case \class{NP}-hard decision problem $L$ to
  inverting a one-way function which is fixed-preimage-sized and surjective, then $\overline{L}$ $\in$
  \class{QIP$(2)$}.
\end{corollary}
\begin{proof}
 This can be done by forcing the prover to give a superposition of all preimages. Suppose the function is k-to-1 and onto, we use the same protocol in Section~\ref{sec:one_reduction} except that the honest prover replies 
\begin{eqnarray}
        |Q^H\rangle_{MV} &=& \frac{1}{\sqrt{2^n}}\sum_{q\in \mathbb{Z}_2^m}
        (\ket{q}\frac{1}{\sqrt{|f^{-1}(q)|}}\sum_{z\in f^{-1}(q)}\ket{z})_M|w_x(q),q\rangle_V\label{eq:q_response_ex}\\
      |T^H\rangle_{MV} &=& \frac{1}{\sqrt{2^n}}\sum_{q\in \mathbb{Z}_2^m}
      (\ket{q}\frac{1}{\sqrt{|f^{-1}(q)|}}\sum_{z\in f^{-1}(q)}\ket{z})_M|0,q\rangle_V.\label{eq:t_response_ex}
\end{eqnarray}
It is not hard to see that $|T^H\rangle_{MV}$ can be mapped to all-zero state by the same unitary in the protocol of Section~\ref{sec:one_reduction}.  Furthermore, since $\ket{Q^H}$ and $\ket{T^H}$ have the same reduced density matrix in $M$, Lemma~\ref{claim:u_indistinguishable} and Lemma~\ref{lemma:maxproj} can be applied. 
\end{proof}

We can extend Corollary~\ref{cor:0} to locally quantum reductions with a quantum-sampling oracle, multiple non-adaptive queries, and smooth-computable distributions.

%-------------------------------------

%----------------------------------------------------------

\section{Oracle separation between \class{coNP} and
  \class{QIP}}\label{sec:conp_qip2}

In this section, we show an oracle $A$ such that coNP$^A \nsubseteq$
QIP($2$)$^A$.
\begin{theorem}
 There exists an oracle $A$ and a language $L(A)\in \class{coNP}^A$ such that $L(A)\notin\class{QIP}(2)^A$.  
\end{theorem} 
\begin{proof}
  We first define the language $L(A)$ on any oracle $A$. For any oracle $A$, let 
  \[
    L(A) = \{1^n:\mbox{ A contains all strings of length n}\}.
  \]
  It is not hard to see that $L(A)\in  \class{coNP}$. Specifically, if $1^n$ is not in $L(A)$, then there exists an $n$-bit string $a$ which is not in $A$, and thus $a$ can be a certificate.    
  
 To show there exists an oracle $A$ such that $L(A)\notin\class{QIP}(2)^A$, we create the oracle $A$ in stages as in~\cite{FS88} via the diagonization technique. We enumerate all possible quantum verifiers in the manner such that $V_i$ is bounded in time by $n^i$, where $n$ is the input size. Then, every verifier $V_i$ will fail to recognize $1^{N_i}$ for some $N_i$ large enough. The main challenge in adapting to the quantum setting is to program the oracle without changing an algorithm's output by too much, even if it can query the oracle in quantum superposition.

 Consider $V_i$, we pick $N_i$ large enough such that $2^{N_i}>12(N_i)^{2i}$ and no oracle queries of length $N_i$ has been made by verifier $V_1,\dots,V_{i-1}$. Note that
 $V_i$ can only make queries with length at most $(N_i)^i$ and at most $(N_i)^i$ queries since the running time is bounded by $(N_i)^i$. Now, every time $V_i$ makes queries which have not been queried before, we let the oracle $A$ output $1$. If there is no prover can convince $V_i$ that $1^{N_i}$ is in $L(A)$ with probability at least $2/3$, then we let A contains all strings with length $N_i$. Otherwise, if there exists a prover which can convince $V_i$ that $1^{N_i} \in L(A)$ with probability at least $2/3$, then there must exist an $N_i$-bit string $x$ such that the sum of its query amplitude is at most $\frac{(N_i)^i}{2^{N_i/2}}$. Finally, we use the hybrid argument as in~\cite{BBBV}. Let $V_i$ be $U_{(N_i)^i+1} A U_{(N_i)^i} A\cdots A U_1$. Let $A'=A$ except that $A'(x) = 0$. Then, for any initial state $\ket{\psi}$
 \begin{align*}
     &\|U_{(N_i)^i+1} A U_{(N_i)^i} A\cdots A U_1\ket{\psi} - U_{(N_i)^i+1} A' U_{(N_i)^i} A'\cdots A' U_1\ket{\psi}  \| \leq \frac{2(N_i)^i}{2^{N_i/2}}.
 \end{align*}
 This implies that the probability that the probability that the same prover convinces $V_i$ with oracle $A'$ is $\frac{2}{3} - \frac{4(N_i)^{2i}}{2^{N_i}} \geq \frac{1}{3}.$
 This contradicts the hypothesis that there is no prover can convince $V_i$ to accept $1^{N_i}$ with probability at most $1/3$.

\end{proof}
\bibliographystyle{plainnat}

%------------------------------------------

\bibliography{avg}

\appendix\label{appendix}

\end{document}